\documentclass{article}
\usepackage[affil-it]{authblk}

\usepackage[english]{babel} %
\usepackage[T1]{fontenc} %
\usepackage[utf8x]{inputenc} %
\usepackage{a4wide} %
\usepackage{palatino} %

\let\bfseriesbis=\bfseries \def\bfseries{\sffamily\bfseriesbis}

\usepackage[a4paper,top=2cm,bottom=2cm,left=2cm,right=2cm,marginparwidth=1.75cm]{geometry}

\usepackage{amsmath,amssymb,mathabx} % extra math symbols
\usepackage{amsthm} % theorems
\theoremstyle{plain}
\newtheorem{theoreme}{Proposition}
\newtheorem{lemma}{Lemma}
\usepackage{graphicx,tikz,tikz-qtree}
\usetikzlibrary{shapes.geometric} % ellipses
\usepackage{enumerate}
\usepackage{todonotes}
\usepackage{linguex} % linguistic supplies
\usepackage{avm} % attribute-value matrix
\usepackage[algoruled,onelanguage,linesnumbered]{algorithm2e} % algorithms
\usepackage{multicol}
\usepackage{hyperref}
\hypersetup{
colorlinks=true, %colorise les liens
breaklinks=true, %permet le retour à la ligne dans les liens trop longs
urlcolor= blue, %couleur des hyperliens
linkcolor= black,	%couleur des liens internes
citecolor=black	%couleur des références
}

\newenvironment{point}[1]%
{\subsection*{#1}}%
{}
\newcommand{\fsbase}[1]{{\scriptsize\normalfont #1}}
\newcommand{\feat}[1]{\textsc{\lowercase{#1}}}
\newcommand{\type}[1]{\text{\textit{#1}}}
\newcommand{\rel}[1]{\text{\textit{#1}}}
\newcommand{\trieq}{\triangleq}
\newcommand{\iseq}{\vDash\!\!\!\!\Dashv}
\newcommand{\eqdef}{\overset{\mathsf{\small{def}}}{=}}
\newcommand{\comp}{\complement\!\cup\!\W}
\newcommand{\WS}{\hat{\W}}
\newcommand{\M}{\mathcal{M}}
\newcommand{\X}{\mathcal{X}}
\renewcommand{\I}{\mathcal{I}}
\newcommand{\W}{\mathcal{W}}

\newcommand{\U}{\mathcal{U}}
\newcommand{\gN}{g_{\mathsf{N}}}
\newcommand{\gW}{g_{\mathsf{W}}}

\newcommand{\svar}{\sigma_{\mathsf{var}}}
\newcommand{\hV}{h_V}
\newcommand{\hW}{h_{\W}}
\newcommand{\hvar}{h_{\mathsf{var}}}

\newcommand{\bl}{\text{\raisebox{0.3\height}{\scalebox{.7}{\ $\blackdiamond$\ }}}} % positionning of the symbol
\newcommand{\gray}[1]{\textcolor{gray}{#1}}
\newenvironment{fs}[1]{\hspace*{\fill}
\begin{tikzpicture}[
x=.22ex, y=.2ex, >=latex, semithick, scale=1, every node/.style={scale=1},
inner sep=.3ex, outer sep=.5ex, minimum size=1.2ex, label distance=-0.8ex]
\tikzset{ft/.style={draw,ellipse,double}};
\tikzset{nod/.style={draw,circle}};
\tikzset{hom/.style={gray,dotted}};
\tikzset{adj/.style={gray,dashed,->}};
#1
}{\end{tikzpicture}\hspace*{\fill}}
\newcommand\REL{\mathsf{Rel}}
\newcommand\NLABEL{\mathsf{Nlabel}}
\newcommand\NVAR{\mathsf{Nvar}}
\newcommand\NNAME{\mathsf{Nname}}
\newcommand\WVAR{\mathsf{Wvar}}
\newcommand\LABEL{\mathsf{Label}}
\newcommand\TYPE{\mathsf{Type}}
\newcommand\ATTR{\mathsf{Attr}}
\newcommand\AVD{\mathsf{AVD}}
\newcommand\AVF{\mathsf{AVF}}
\newcommand\WR{\mathsf{wr}}
\newcommand\INWR{\mathsf{inwr}}
\newcommand\IN{\mathsf{in}}
\newcommand\FACTS{\mathsf{Facts}}
\newcommand{\dom}{\mathsf{dom}}
\newcommand{\codom}{\mathsf{codom}}
\newcommand{\FORM}{\mathsf{form}}
\newcommand{\ATOMS}{\mathsf{ATOMS}}
\newcommand{\mea}[1]{|#1|_{\ATTR}}

\renewcommand{\phi}{\varphi}
\renewcommand{\epsilon}{\varepsilon}

\begin{document}

\title{Introduction of Quantification in Frame Semantics}

\author{Valentin D. Richard\\
{\tt valentin.richard@ens-paris-saclay.fr}\\
under the supervision of Laura Kallmeyer and Rainer Osswald\\ DFG Collaborative Research Centre 991\\ Heinrich-Heine Universit\"at, D\"usseldorf (Germany)}

\date{from February 25, to August 16, 2019}

\affil{M1 internship report\\ENS Paris-Saclay, computer science department}

\maketitle

\begin{abstract}
Feature Structures (FSs) are a widespread tool used for decompositional frameworks of Attribute-Value associations. Even though they thrive in simple systems, they lack a way of representing higher-order entities and relations. This is however needed in Frame Semantics, where semantic dependencies should be able to connect groups of individuals and their properties, especially to model quantification. To answer this issue, this master report introduces wrappings as a way to envelop a sub-FS and treat it as a node. Following the work of [Kallmeyer, Osswald 2013], we extend its syntax, semantics and some properties (translation to FOL, subsumption, unification). We can then expand the proposed pipeline. Lexical minimal model sets are generated from formulas. They unify by FS value equations obtained by LTAG parsing to an underspecified sentence representation. The syntactic approach of quantifiers allows us to use existing methods to produce any possible reading. Finally, we give a transcription to type-logical formulas to interact with the context in the view of dynamic semantics. Supported by ideas of Frame Types, this system provides a workable and tractable tool for higher-order relations with FS.
\end{abstract}

\textbf{Keywords:} feature structure, semantic frame, quantification, attribute-value logic, minimal model, syntax-semantics interface

\setcounter{page}{-1}
%%%%%
\section*{Summary sheet}
\begin{point}{General context}
  
%   De quoi s'agit-il? D'o? vient la question? Quels sont les travaux
%   d?j? accomplis dans ce domaine dans le monde?

My internship falls within the domain of computational linguistics, the interdisciplinary field aiming at modelling natural language with computational tools. More precisely the parts of linguistics that I considered were semantics and the syntax-semantics interface, i.e. the way that meanings of words compose with respect to the structure of a sentence to form a sentential meaning. I followed a theoretical approach, using mainly logic and formal grammar methods to model the studied phenomenon: quantification.

Quantification is a feature of natural language raised among others by determiners (every, some, a/an, no, the, most,...). Since \cite{BarwiseCooper:81} the principal mechanisms of this phenomenon are well understood and well connected to formal logic. Their approach uses the widespread semantics à la Montague \cite{Montague:73}, i.e. predicate-logical (typed $\lambda$-calculus) formulas to represent the meanings of sentences.

However, recent studies advocate for other modellings of meaning. Semantic frames, dealing with the concept of mental representation, are an example of these and are namely supported by philosophical and cognitive results \cite{Barsalou:92,Loebner:14}. Formal specifications \cite{KallmeyerOsswald:13} have already been proved able at connecting this formalism to other famous ones, likes Tree-Adjoining Grammars for syntax.

\end{point}

\begin{point}{The studied problem}
  
%   Quelle est la question que vous avez r?solue? Pourquoi est-elle
%   importante, ? quoi cela sert-il d'y r?pondre?  Pourquoi ?tes-vous
%   le premier chercheur de l'univers ? l'avoir pos?e?

The question is: how to model the phenomenon of quantification in semantic frames? Previous attempts \cite{Hegner:93,KallmeyerOsswaldPogodalla:17} tried to answer this with external tools acting on it. Even if such peripheral methods seem useful and relevant from a computer science viewpoint, the main interest in giving quantifier a full-fledged place in frames is also to provide a cognitive idea of them. In other words: how do human beings represent quantification in their minds?

The DFG Collaborative Research Centre 991 is one of the most interested in this question because their members are studying mental representation models of natural language meaning. Namely, my supervisors Laura Kallmeyer and Rainer Osswald were the most likely to carry out a research project on that topic thanks to their career specialized not only in semantics but also in the syntax-semantics interface and logical formalisms. A previous analysis \cite{KallmeyerRichter:14} do not make clear distinctions between some representation levels. It shows that the subject can be delicate, especially if we look for elegance additionally to convenience and computability.

Finally, it must be borne in mind that the study of quantification is a milestone in the development of a formal NLP (natural language processing) system, as almost every sentence contains determiners. The TreeGraSP project related to the institute is currently designing a program to support the theoretical work. Therefore, they need an implementable and robust way to cope with quantification.

\end{point}

\begin{point}{The proposed contribution}

%   Qu'avez vous propos? comme solution ? cette question? Attention, pas
%   de technique, seulement les grandes id?es! Soignez particuli?rement
%   la description de la d?marche \emph{scientifique}.

My contribution is an extension of Feature Structures, the formal system in which frames are defined, called Wrappings. This new tool allows us to deal with Noun Phrase quantified sentences. The system is developed in a close idea of \cite{KallmeyerRichter:14}. Firstly, the syntactic structure is captured by a grammar, where each word is coupled with a logical specification transformed into a minimal model. The models are then unified with respect to the parsing, forming so-called Quantified Complexes. This complexes can then interact with the current knowledge of the world by verifying or updating it with this new information.

The main representation was inspired by the idea of Frame Types as in \cite{BaloghOsswald:19} and the surrounding structure was constructed to fit this goal, with the help of existing tools (e.g. \cite{KollerThater:05}).
 
\end{point}

\begin{point}{Arguments in favor of its validity}

%   Qu'est-ce qui montre que cette solution est une bonne solution? Des
%   exp?riences, des corollaires? Commentez la \emph{stabilit?} de votre
%   proposition: comment la validit? de la solution d?pend-elle des
%   hypoth?ses de travail?

The proposed model enjoys a simple and intuitive representation of quantified sentences. Moreover, the basic structure of quantifiers, namely scope constraints, is rendered easily and allowing underspecification. This last feature is cognitively very relevant, for we often do not know which particular reading of a sentence was meant. Proofs also back up the correct treatment of the involved models, from their generations to their dynamic processes.

However, it is still not clear if this method is sufficient for every type of quantification, namely for adverbial quantification or attitude verbs.

\end{point}

\begin{point}{Assessment and perspectives}
  
%   Et apr?s? En quoi votre approche est-elle g?n?rale? Qu'est-ce que
%   votre contribution a apport? au domaine? Que faudrait-il faire
%   maintenant? Quelle est la bonne \emph{prochaine} question?

This work hopes to provide a strong enough representation of quantification in frame-like modellings so that it could be used in other analyses and implemented in programs. In any case, it is only a starting point, and further developments of quantifying structures in frames should follow. For example, a remaining question is: how can we model the particular interaction of dot objects (words having an intrinsic and variable ambiguity, like \textit{book}: physical object or content) with quantification?

\end{point}

\newpage
\setcounter{page}{1}

\tableofcontents

%%%%%
\section{Context}

%%%
\subsection{The DFG Collaborative Research Centre 991}

The \href{https://frames.phil.uni-duesseldorf.de/}{DFG Collaborative Research Centre 991} (CRC991) is a research project funded by the German Research Community (Deutsche Forschungsgemeinschaft). The center is supported by and based in the buildings of the \href{https://www.uni-duesseldorf.de/home/startseite.html}{Heinrich-Heine Universität of Düsseldorf} (HHU) and is attached to the faculty of humanities. The HHU is one of the youngest higher education institutions in the state of North Rhine-Westphalia, in western Germany. It includes 5 faculties all located in the south of Düsseldorf, harboring 35,000 students, more than 2,000 lecturers and 900 further employees. The diversity of social behaviours is specially protected and celebrated by the institution. The student life is rich and encouraged by the university itself. Research is also well promoted, as the diverse collaborative centers attest it.

The CRC991 gathers 56 members, including 16 PIs (professors) and 18 PhD students, all divided into 10 active projects. They organize various events and seminars, that I had the chance to attend. The SToRE Graduate Training Program is an intensified model of doctoral supervision guiding PhD students in their professional career. Being among them helped me better apprehend how this period of life looks like.

The topic uniting the members is the structure of representations in language, cognition, and science. They especially work on the notion of concepts and mental representations (representation of the world in the mind), under the theory first advocated by Lawrence Barsalou \cite{Barsalou:92}: semantic frames. This is the starting point of different studies concerning a formal approach of the philosophical notion of concept, their psychological relevance, the meaning conveyed by natural language and its connection with syntax. This last axis is also linked to the Fillmore's frame program \cite{Fillmore:82}, aiming at capturing the relations between semantic and syntactic dependencies. The unsupervised frame induction project uses the corresponding developed formalism.

%%%
\subsection{The TreeGraSP project}

\href{https://treegrasp.phil.hhu.de/}{TreeGraSP} is a research project funded by an ERC Consolidator Grant and led by Professor Doctor Laura Kallmeyer, being also the director of the CRC991 and my supervisor. Its goal is to provide a statistical grammar (Role and Reference Grammar or Tree-Adjoining Grammar) suited for semantic parsing (with frames) and to develop adapted resources (\href{http://xmg.phil.hhu.de/}{XMG compiler} and \href{https://github.com/spetitjean/TuLiPA-frames}{TuLiPA parser}).

My internship is affiliated to the A02 project of the CRC991: ``Argument linking and extended locality: A frame-based implementation''. However my work is a bit peripheral in the sense that a modelling of quantification in frames has been expected and thought of for a long time, but it does not exactly fit in the topic of the current funding period.

%%%%%
\section{Research topic}

%%%
\subsection{The concept of frames}

\subsubsection{A theory of Natural Language semantics}

\textbf{Barsalou's frames} \cite{Barsalou:92} are a theorized notion of mental representations of the world with an aspiration to cognitive adequacy \cite{Loebner:14}. They are made of concepts which are connected by attributes. The functionality (i.e. that the relations are functions) of the latter is fundamental. In natural language semantics the \textbf{basic concepts} are individuals (\type{cat}, \type{Sarah}, \type{city},...), events (\type{seeing}, \type{stroking}, \type{spying},...) and properties (\type{white}, \type{tall},...). Attributes are \textbf{thematic roles} (i.e. semantic dependencies), like \feat{agent} (deliberately performs the action), \feat{theme} (undergoes the action but does not change its state), \feat{manner} (the way in which an action is carried out), \feat{size} (size of the individual), etc. The list is wide and subtle meaning discrepancies are not investigated here, except quickly in paragraph \ref{subsub:constraints}. The main concept expressed is called the center frame. This way, the global frame can be seen as smaller frames connected to one another.

In this sense, frames extend the neo-Davidsionan approach to semantics \cite{Parson:90}. This theory stipulates that, contrary to the usual verbal predicates, events should be split into a full-fledged event concept, which the different semantic dependencies are connected to by thematic roles. For example, instead of a predicate $\type{stroke}(x,y)$, the right decomposition is $\exists e.~\type{stroking}(e)\wedge \feat{agent}(e,x)\wedge \feat{theme}(e,y)$. This allows us to add as many optional dependencies as needed, adverbial modifiers and also to represent the same way related lexical items having just different dependencies (e.g. an intransitive version of a transitive verb does not have an object). Moreover different syntactic structures conveying the same meaning (e.g. active and passive voice sentences) can also be represented the same way.

\subsubsection{Feature Structures}

Frames are usually formalized as \textbf{Feature Structures} (FS). They consist of values, which can be either atomic or a set of couples of a feature and a value. Set $\TYPE$ a finite set a base types, $\ATTR$ a finite set of attributes and $\NLABEL$ a set of labels. A Feature Structure Value is defined as a triplet of the form $FV::= (L,T,A)$, with $L\subseteq\NLABEL$, $T\subseteq \TYPE$ and $A\subseteq\ATTR\times FV$, $A=\emptyset$ being the base case. A Feature Structure (FS) is then a nonempty set of Feature Structure Values where all label sets $L$ are nonempty. This last constraint is known as the reachability constraint, i.e. each subFS must be attribute-accessible from a labelled FS Value.

A Feature Structures can be represented with an Attribute-Value Matrix (e.g. Fig.~\ref{avm-fs:tall-daughter} on the left), where the features (here attributes) are depicted as a list with their values on the right, and the types at the top. The center frame is put at the top of the list. Another depiction is as a labelled and tagged directed graph (e.g. Fig.~\ref{avm-fs:tall-daughter} on the right). Nodes are FS Values, with labels inside and types tagging them in italic. Edges tagged by features in small caps connect FS Values. The center node is indicated by a rectangle node. Moreover, each present label labels a unique node, so this depiction cannot hold redundancy. As it also allows to better perceive relations between distant concepts in a big frame at one glance, only this depiction style will be used in the following.

\begin{figure}[ht]
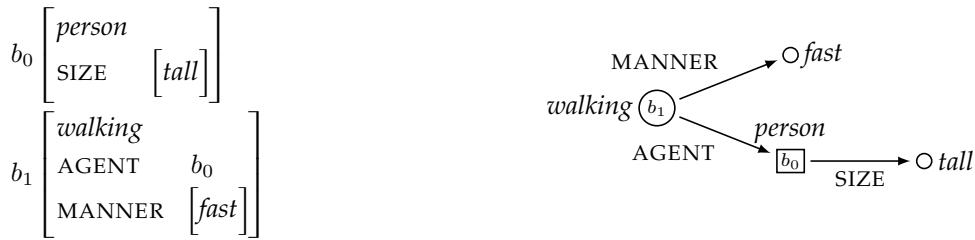

\begin{minipage}{0.35\textwidth}
\centering
\begin{tabular}{l}
\begin{avm}
$b_0$ \[ \type{person} & \\ \feat{size} & \[ \type{tall} \] \]
\end{avm} \\[-1ex]
\begin{avm}
$b_1$ \[ \type{walking} & \\ \feat{agent} & $b_0$ \\
    \feat{manner} & \[ \type{fast} \] \]
\end{avm}
\end{tabular}
\end{minipage}
\begin{minipage}{0.6\textwidth}
\begin{fs}
\draw (0pt, -20pt)node[draw, rectangle, label={90:\type{person}}](0){\fsbase{$b_0$}};
\draw (-50pt, 0pt)node[nod, label={180:\type{walking}}](1){\fsbase{$b_1$}};
\draw (0pt, 20pt)node[nod, label={0:\type{fast}}](2){};
\draw (50pt, -20pt)node[nod, label={0:\type{tall}}](3){};
\path[->] (1) edge node[below left]{\feat{agent}}(0);
\path[->] (1) edge node[above left]{\feat{manner}}(2);
\path[->] (0) edge node[below]{\feat{size}}(3);
\end{fs}
\end{minipage}
\caption{Examples of frames (Feature Structures) for the expression \textit{tall fast walker}, as an Attribute-Value Matrix on the left and as an equivalent directed graph on the right}
\label{avm-fs:tall-daughter}
\end{figure}

On Fig.~\ref{avm-fs:tall-daughter} we can already see how frames help to better control the semantics of words. The term \type{walker} can be split into two concepts: a \type{walking} event which \feat{agent} is a \type{person}. This explains how the adjectives \textit{tall} and \textit{fast} account for different types of contributions to the noun \textit{walker}, namely by being attached to the \type{walking} or the \type{person} concept respectively. A good syntax-semantics interface formalism should be able to capture this subtlety.

Feature Structures can be compared to each other by \textbf{subsumption}. We say that a FS subsumes another one intuitively if the second is obtained from the first by adding information. The subsumed frame can have new disjoint concepts, newly connected ones, additional details or identification of concepts. It is performed by types, labels, attributes or FS Values addition. In directed graphs identification of concepts leads besides to node merging.

Subsumption can be defined by the existence of a labelled and tagged directed graph homomorphism: a directed graph homomorphism respecting types, attributes and labels. An example is depicted in Fig.~\ref{fs:hom-cat} with the node mapping in gray dotted lines. Here, the fact that a \type{cat} is an \type{animal} is not taken into account. But such a type hierarchy (ontology) could be directly implemented in the FS, instead of being introduced by constraints (see. paragraph \ref{subsub:constraints}).

\begin{figure}[ht]
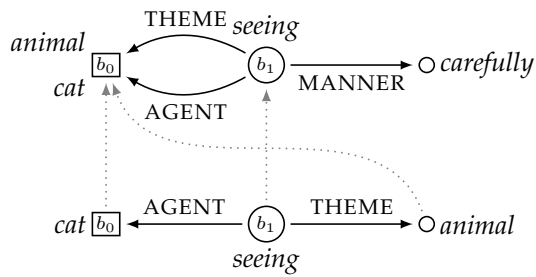

\begin{fs}
\begin{scope}
\draw (0pt, 0pt)node[draw, rectangle, label={180:\type{cat}}](11){\fsbase{$b_0$}};
\draw (60pt, 0pt)node[nod, label={-90:\type{seeing}}](10){\fsbase{$b_1$}};
\draw (120pt, 0pt)node[nod, label={0:\type{animal}}](12){};
\path[->] (10) edge node[above]{\feat{agent}}(11);
\path[->] (10) edge node[above]{\feat{theme}}(12);
\end{scope}

\begin{scope}[yshift=60]
\draw (0pt, 0pt)node[label={160:\type{animal}}](22){\fsbase{$\textcolor{white}{b_0}$}};
\draw (0pt, 0pt)node[draw, rectangle, label={200:\type{cat}}](21){\fsbase{$b_0$}};
\draw (60pt, 0pt)node[nod, label={90:\type{seeing}}](20){\fsbase{$b_1$}};
\draw (120pt, 0pt)node[nod, label={0:\type{carefully}}](23){};
\path[->] (20) edge[bend left] node[below]{\feat{agent}}(21);
\path[->] (20) edge[bend right] node[above]{\feat{theme}}(22);
\path[->] (20) edge node[below]{\feat{manner}}(23);
\end{scope}

\path[->] (10) edge[hom] (20);
\path[->] (11) edge[hom] (21);
\path[->] (12) edge[hom,out=110,in=-70] (21);
\end{fs}
\caption{Two Feature Structures: \textit{cat seeing an animal} at the bottom subsuming \textit{cat carefully seeing itself} at the top}
\label{fs:hom-cat}
\end{figure}

Another useful operation between Feature Structures is \textbf{unification}. The unification of two FSs, when it exists, is the minimal FS carrying the information of both inputs. Here, it simply consists of putting both FSs next to each other and merging the nodes having the same label. So it corresponds to the least upper bound with respect to (w.r.t.) subsumption. Fig.~\ref{fs:unif-John} shows an example of unification. It may be remarked that these frames could be refined by splitting \type{Mary} into \type{person} of \feat{name} \type{Mary} (and same for \type{John}), but we adopt a more concise notation, for it is not the main point. From now on center nodes will not be represented any longer because we are considering not only individuals but whole phrases.

\begin{figure}[ht]
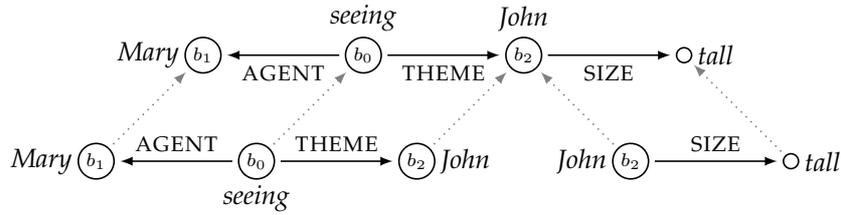

\begin{fs}
\begin{scope}[xshift=-40]
\draw (0pt, 0pt)node[nod, label={180:\type{Mary}}](11){\fsbase{$b_1$}};
\draw (60pt, 0pt)node[nod, label={-90:\type{seeing}}](10){\fsbase{$b_0$}};
\draw (120pt, 0pt)node[nod, label={0:\type{John}}](12){\fsbase{$b_2$}};
\path[->] (10) edge node[above]{\feat{agent}}(11);
\path[->] (10) edge node[above]{\feat{theme}}(12);
\end{scope}

\begin{scope}[xshift=160]
\draw (0pt, 0pt)node[nod, label={180:\type{John}}](22){\fsbase{$b_2$}};
\draw (60pt, 0pt)node[nod, label={0:\type{tall}}](23){};
\path[->] (22) edge node[above]{\feat{size}}(23);
\end{scope}

\begin{scope}[yshift=40]
\draw (0pt, 0pt)node[nod, label={180:\type{Mary}}](31){\fsbase{$b_1$}};
\draw (60pt, 0pt)node[nod, label={90:\type{seeing}}](30){\fsbase{$b_0$}};
\draw (120pt, 0pt)node[nod, label={90:\type{John}}](32){\fsbase{$b_2$}};
\draw (180pt, 0pt)node[nod, label={0:\type{tall}}](33){};
\path[->] (32) edge node[below]{\feat{size}}(33);
\path[->] (30) edge node[below]{\feat{agent}}(31);
\path[->] (30) edge node[below]{\feat{theme}}(32);
\end{scope}

\path[->] (10) edge[hom] (30);
\path[->] (11) edge[hom] (31);
\path[->] (12) edge[hom] (32);
\path[->] (22) edge[hom] (32);
\path[->] (23) edge[hom] (33);
\end{fs}
\caption{Three Feature Structures: \textit{Mary sees John} at the bottom on the left, \textit{John is tall} at the bottom on the right and their unification \textit{Mary sees John and John is tall} at the top} 
\label{fs:unif-John}
\end{figure}

\subsubsection{Attribute-Value Logic}

The idea of coupling attributes (or features) with (recursive) values is a general logical tool. It has been used extensively in the literature of computational linguistics, especially for constraint-based grammars (Head-driven Phrase Structure Grammar, Lexical Functional Grammar). This kind of formalism aiming at capturing the syntax of Natural Language posits that every sentence is possible, except those infringing a certain set of generated constraints. \textbf{Attribute-Value Logic} (AVL) \cite{Johnson:88} was developed as a description logic to express Attribute-Value grammars. The axiomatisation enjoys interesting properties of completeness, soundness and compactness. Well-formed formulas can be mapped to quantifier-free First Order Logic with equality and function symbols, proving the NP-completeness of its satisfiability problem.

Working in the framework of categories has been investigated by using topological notions applied to logic \cite{Osswald:99}. Thanks to this, theoretical transformations between formulas (opens), proofs (set of closed sets) and fact-sets (points) can be produced. The category of fact-sets is then showed to be isomorphic to the one of feature structure.
AVL was recently reused to specify frames \cite{KallmeyerOsswald:13}. Here, the point is rather to describe what are the concepts at stake, and how they are connected. Hybrid Logic (a modal logic) was also tested \cite{KallmeyerOsswaldPogodalla:17}, because it figures out to be relevant when talking about graph-like structures and is already linked to other grammatical systems. But it lacks the fundamental functionality of attribute, without bringing more computational power or convenience.

The core of any AVL is made of two sorts: the descriptions $\AVD$ and the formulas $\AVF$. The descriptions are terms specifying what holds at a certain node. Formulas, on the contrary, tell what happens in the general frame.

\begin{equation}
\label{eq:avl-core}
\begin{array}{rrl}
\AVD:& \phi,\psi ::=& \feat{P}:\phi~|~t~|~\phi\wedge\psi~|~\top\\
\AVF:& \chi,\xi  ::=& k\cdot\phi~|~k\cdot p\trieq l\cdot q~|~\chi\wedge\xi~|~\top
\end{array}
\end{equation}

\noindent with $\feat{p}$ an attribute, $p$ and $q$ words of attributes, $t$ a type and $k$ and $l$ labels. An $\AVD$ can be considered encoding FS Values, and $\trieq$ stands for path equality. We can enrich it with relations if wanted and add syntactical sugar. Equation \eqref{eq:Mary} provides an example of formula for the sentence \ref{ex:Mary}.

\ex. \label{ex:Mary} Mary walks fast.

\begin{equation}
\label{eq:Mary}
b_0\cdot(\type{walking}\wedge\feat{agent}:\type{Mary}\wedge\feat{manner}:\type{fast})
\end{equation}

The key idea is to create a \textbf{minimal model} out of an Attribute-Value Formula. It has to be understood as the minimal information that one has in mind at the utterance of a sentence, which semantics is described by such a formula. From a practical point of view, AVL is an easy way to specify the meanings of words in a lexicon whereas Feature Structures are better to manipulate at a large scale.

%%%
\subsection{Lexicalized Tree-Adjoining Grammars with features}

Thanks to their simplicity, Feature Structures have been employed in lots of formalisms of computational linguistics. One of the first uses is in grammars to specify grammatical features. For example, the agreement of the subject (e.g. $\feat{person}:\type{1st}$, $\feat{number}:\type{singular}$) have to be the same as the agreement of the verb, or more exactly they must be unifiable (e.g. \textit{see} can be conjugated to the first or second person of singular, or plural). Some tools, like the \href{https://www.nltk.org/}{NLTK python library}, already use an implemented version with Context-Free Grammars.

However, CFGs are well known to be unable to capture long-distance phenomena \cite{Kallmeyer:10} present in some languages, like Dutch. That is why \textbf{Lexicalized Tree-Adjoining Grammars with features} (LTAG), which are mildly context-sensitive, are used instead to account for the syntactic structure of Natural Languages. As this formalism is widely spread, its presentation here is just a quick non-formal overview to be able to follow the rest of the report. For further details, we refer to \cite{Kallmeyer:10}.

LTAGs are close to CFG, but each word (lexical item) has an own elementary tree. There are two kinds of it: initial trees and auxiliary trees. The latter have exactly one foot node (indicated via an asterisk) which symbol is identical to that of the root. Two operations permit to combine these trees. Substitution consists in attaching an initial tree to a non-terminal node of another tree if the symbols match. Adjunction consists in inserting a whole auxiliary tree in another tree at a place where the symbols match. The lower part of the second tree is put at the foot node of the auxiliary tree.

Features Structures can be added to every non-terminal node at the top of at the bottom of it \cite{KallmeyerRomero:08}, usually depicted as little Attribute-Value Matrices. During substitution and adjunction, FSs at the same place has to unify and eventually the top and the bottom FSs of any node also. This yields label equations. Intuitively, this enables information transmission through the different elementary trees. An example is depicted in Fig.~\ref{tag:Mary}.

\begin{figure}[ht]
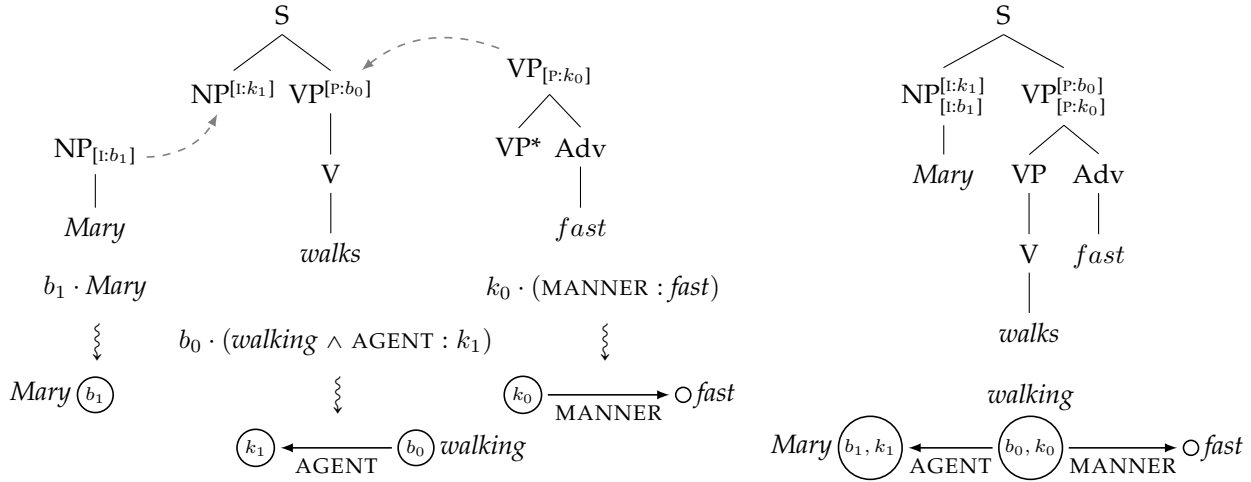

\begin{fs}
\begin{scope}[shift={(-70pt,-50pt)}]
\Tree 
[.\node(np1){NP$_{\text{[{\sc i}:$b_1$]}}$}; \type{Mary} ]
\end{scope}
\begin{scope}[shift={(-70pt,-100pt)}]
\draw (0,0) node {$b_1\cdot \type{Mary}$};
\draw (0,-20pt) node {\rotatebox{-90}{$\rightsquigarrow$}};
\draw (0,-40pt) node[nod,label={180:\type{Mary}}]{\fsbase{$b_1$}};
\end{scope}

\begin{scope}
\Tree 
[.S \node(np0){NP$^{\text{[{\sc i}:$k_1$]}}$}; [.\node(vp0){VP$^{\text{[{\sc p}:$b_0$]}}$}; [.V \textit{walks} ] ] ]
\end{scope}
\begin{scope}[shift={(20pt,-120pt)}]
\draw (0,0) node {$b_0\cdot(\type{walking}\wedge\feat{agent}:k_1)$};
\draw (0,-20pt) node {\rotatebox{-90}{$\rightsquigarrow$}};
\draw (30pt,-40pt) node[nod,label={0:\type{walking}}](00){\fsbase{$b_0$}};
\draw (-30pt,-40pt) node[nod](01){\fsbase{$k_1$}};
\path[->] (00) edge node[below]{\feat{agent}} (01);
\end{scope}

\begin{scope}[shift={(100pt,-20pt)}]
\Tree 
[. \node(vp2){VP$_{\text{[{\sc p}:$k_0$]}}$}; VP* [.Adv $fast$ ] ]
\end{scope}
\begin{scope}[shift={(120pt,-100pt)}]
\draw (0,0) node {$k_0\cdot(\feat{manner}:\type{fast})$};
\draw (0,-20pt) node {\rotatebox{-90}{$\rightsquigarrow$}};
\draw (-30pt,-40pt) node[nod](00){\fsbase{$k_0$}};
\draw (30pt,-40pt) node[nod,label={0:\type{fast}}](02){};
\path[->] (00) edge node[below]{\feat{manner}} (02);
\end{scope}

\draw[adj,bend right] (np1) to (np0);
\draw[adj,bend right] (vp2) to (vp0);

\begin{scope}[shift={(270pt,0pt)}]
\Tree 
[.S [.NP$^{\text{[{\sc i}:$k_1$]}}_{\text{[{\sc i}:$b_1$]}}$ \type{Mary} ]
    [.VP$^{\text{[{\sc p}:$b_0$]}}_{\text{[{\sc p}:$k_0$]}}$
        [.VP [.V \textit{walks} ] ]
        [.Adv $fast$ ] ] ]
\end{scope}
\begin{scope}[shift={(250pt,-120pt)}]
\draw (30pt,-40pt) node[nod,label={90:\type{walking}}](00){\fsbase{$b_0,k_0$}};
\draw (-30pt,-40pt) node[nod,label={180:\type{Mary}}](01){\fsbase{$b_1,k_1$}};
\path[->] (00) edge node[below]{\feat{agent}} (01);
\draw (90pt,-40pt) node[nod,label={0:\type{fast}}](02){};
\path[->] (00) edge node[below]{\feat{manner}} (02);
\end{scope}
\end{fs}
\caption{Example of LTAG parsing of the sentence \ref{ex:Mary}. On the left, the \textit{Mary} tree substitutes at the subject position of the \textit{walks} tree, and the \textit{fast} tree adjoins at the VP node of \textit{walks}. Below each elementary tree the associated AVL semantic specification is transformed ($\rightsquigarrow$) to a minimal model below it. On the right the resulting tree, which FS Value equations ($k_1\trieq b_1$ and $k_0\trieq b_0$) trigger the unification of lexical frames.}
\label{tag:Mary}
\end{figure}

Following the traditional syntactical analysis of English, we use part-of-speech (morpho-syntactic categories) tags as non-terminal symbols: S (sentence), N (noun), NP (noun phrase), V (verb), VP (verb phrase), Adv (adverb), P (preposition), PP (prepositional phrase) and D (determiner). Terminal nodes (or anchors) are the lexical items. Features used are usually grammatical features. Here, as we will not focus on the exact syntax, we do not represent them. But we employ semantic features: \textsc{i} (individual) and \textsc{p} (predicate).

The great contribution of \cite{KallmeyerOsswald:13} is to couple each elementary tree with a semantic part, specified in AVL and use the same labels as in the tree. Thus the semantic realizations of syntactic dependencies are caught (e.g. the \feat{agent} of \type{walking} is the subject of \textit{walks}). This idea to connect grammatical functions and thematic roles has been studied for a while in the theory of \textbf{Fillmore's frames} \cite{Fillmore:82}. The \href{https://framenet.icsi.berkeley.edu/fndrupal/}{FrameNet} tool is developed with dependency grammars and lexical semantic FSs. It is another major source of inspiration in the computational-linguistics part of the CRC991.

The minimal model construction can be processed just after the semantic specification and stored in a lexicon. The derivation in LTAGs induces label equations, with respect to which the lexical models have to unify to produce the sentential final frame.

%%%
\subsection{Quantification}

\subsubsection{Quantification phenomena in Natural Language}

In Natural Language (NL) \textbf{quantification} is a phenomenon raised by determiners (\textit{a/an, some, every, each, all, the, most, few, little, two,...} or complex determiner \textit{a lot of, less than four, a finite number of,...}) and specifying the quantity of individuals of a certain kind involved in a predicate. The common syntactical analysis is the following: a \textbf{determiner} $D$ in a Noun Phrase (NP) has the rest of the NP (called noun nucleus) as its \textbf{restrictor} and a part of the main Verb Phrase (VP) as its \textbf{nuclear scope}. The determiner and its restrictor form a quantifier. Some quantifiers have an equivalent in classical logic, like $\forall x\in\theta$ for \textit{every} $\theta$ or $\exists x\in\theta$ for \textit{some} $\theta$. Formulas in \ref{eq:bark-dog} show how sentence \ref{ex:bark-dog} can be turned into a type-logical semantic formula.

\exi. \a.\label{ex:bark-dog} [NP [D Every ] [N dog ]Restrictor ] [VP [V barks ] ]Nuclear-Scope .
    \b.\label{eq:bark-dog} $\forall x\in\type{dog},~x\in\type{barking}$ \hfill (Set Theory style)\\
    $\rel{every}(x,\type{dog}(x),\type{barking}(x))$ \hfill (modern GQT style)\\
    $\forall x.\type{dog}(x)\to(\exists e.~\type{barking}(e)\wedge\feat{agent}(e,x))$ \hfill (FOL style with neo-Davidsonian decomposition)
    
However the distinction between the determiner and its restrictor is not common in formal logic. The General Quantifier Theory (GTQ) \cite{BarwiseCooper:81} was one of the first to provide a representation of quantification in logic. It posits that the restrictor and the nuclear scope represent two sets $\theta$ and $\delta$ (or $\eta$-expanded $\{x~|~\theta(x)\}$ and $\{x~|~\delta(x)\}$) respectively, and the interpretation of the determiner is a relation $D$ between both. For example $\rel{every}~\theta~\delta$ if and only if (iff) $\theta\subseteq\delta$, or $\rel{some}~\theta~\delta$ iff $\theta\cap\delta\neq\emptyset$.

The fundamental property of NL quantifiers is conservativity, i.e. for every $\theta,\delta$, $D~\theta~\delta$ iff $D~\theta~(\delta\cap\theta)$. For \ref{ex:bark-dog} it means that \textit{every dog is a dog who barks}. This namely allows us to only work within the restrictor set. Other properties, like monotonicity w.r.t. the first or second argument, lead to a classification of quantifiers.

The restrictor and nuclear scope are called scopes of the determiner. Whereas the restrictor is usually well-defined as the noun nucleus, the nuclear scope can be ambiguous in an utterance. Especially when the sentence encompasses several quantifiers, most of the time it is unclear which one scopes over (or outscopes) the other, i.e. which one is in the scope of the other. This yields ambiguity, with several possible readings (unambiguous resolutions). This phenomenon is directly linked with the ordering of quantifiers in logic. For instance, in \ref{ex:stroking} both \textit{a} and \textit{every} scope over the verb. The latter is the minimal predicate that they can have in their nuclear scope. So the sentence is ambiguous between the \ref{eq:stroking-some} and the \ref{eq:stroking-every} readings. 

\exi. \a.\label{ex:stroking} [NP A [N student ] ] [VP [V stroked ] [NP every [N cat ] ] ] .
    \b. \label{eq:stroking-some} $\exists x.~\type{student}(x)\wedge (\forall y.~\type{cat}(y)\to\type{stroking}(x,y))$ \hfill ($\rel{a}>\rel{every}$ reading)
    \c. \label{eq:stroking-every} $\forall y.~\type{cat}(y)\to (\exists x.~\type{student}(x)\wedge\type{stroking}(x,y))$ \hfill ($\rel{every}>\rel{a}$ reading)

This property of possessing one or several scopes and inferring an idea of quantification over some type of variable is also shared by other kinds of words. Some adverbs attached to verbs (\textit{possibly, certainly, always, never, sometimes,...}) or attitudes verbs (\textit{think, believe,...}) are typically able to introduce some quantification over possible worlds, time or location. Nevertheless, they work a bit differently from NP quantifiers. As this work is a first jump into quantification in frames, we will only focus on determiners, and exactly on simple and logical determiners: \textit{every}, \textit{no} and \textit{a/an}\footnote{\textit{a/an} is approximated here to the same meaning as \textit{some}, i.e. only existential}.

\subsubsection{Computational-linguistics treatment of NP quantifiers}

Modelling NP quantifiers has lead to lots of approaches, but the majority falls under the scope of $\lambda$-calculus with typed logic. Content and function words are represented by terms or operators, and the meaning is a truth-conditional (or dynamic semantic) formula \cite{Montague:73,BeyssonBlindDeGrooteGuillaume:17}. The strategies to derive the different possible reading from the surface parsing (e.g. quantifier rising) are diverse and depending on the grammar used. The most promising theories suggest an \textbf{underspecified representation} of the sentence (e.g. Minimal Recursion Semantics \cite{CopestakeFlickingerSagPollard:05}). In this shape, the formula is not entirely built, but \textbf{scope constraints} specify how parts behave and can be arranged. Hole Semantics inspired Dominance Graphs \cite{KollerThaterPinkal:10}, which provide efficient algorithms to compute the set of possible readings \cite{KollerThater:05}. Interesting results also concerns the computation of the set of (approximated) weakest readings \cite{KollerThater:10}, the readings which logically entail every other one.

Other kinds of semantic theories, such as model-theoretic ones, cannot still represent quantifiers as part of the model. Indeed this phenomena requires second-order entities and relations between them. Graph-like modellings like semantic nets \cite{Helbig:06} just represent quantifiers as nodes and edges without any other interpretation. However, they record their functioning by special attributes (determination, variability, cardinality,...) and distinguish collective, cumulative or distributive readings (not investigated here). These features are now being developed in the quantification annotation project of the \href{https://www.iso.org/obp/ui/#iso:std:iso:24617:-6:ed-1:v1:en}{ISO Semantic Annotation Framework} \cite{Bunt:19}, aiming at providing rich enough durable conventions for NL annotation. Another promising approach is the recent incorporation of quantification in Vector Semantics \cite{HedgesSadrzadeh:19}. Such function words receive a category interpretation, which is consistent with the functioning of determiners but independent from the word. Work is still in progress to try to learn a distributional meaning for each quantifier, as well as entailment properties. Let us finally cite the Discourse Representation Theory \cite{KampGenabithReyle:10}, which provides a simple and powerful tool, especially good to handle discourse-level phenomena, like quantification, anaphora and donkey sentences\footnote{A donkey sentence is a sentence of the kind \textit{\label{fn:donkey} Every farmer who owns a donkey$_i$ loves it$_i$.}, where the quantifier \textit{every} raises the individual \textit{donkey}$_i$ which can be retrieved (by \textit{it}$_i$) outside the scope of the quantifier \textit{a} introducing it.}. It must be born in mind that, given that determiners are prevalent in NL, quantification modelling is a mandatory milestone for every semantic or syntax-semantics-interface theory.

Several various attempts of introducing quantification for frames show that there is no canonical way of doing so. In the tradition of type-logical analyses, we can produce quantified formulas outside the model \cite{KallmeyerOsswaldPogodalla:17}, which describe how it is (or should be). A more exotic approach consists in considering quantification on the set of frame ordered by subsumption \cite{Muskens:13}. Yet, the idea that a Feature Structure may be also viewed as the set of FSs subsumed by it is not new \cite{Hegner:93}. But a real implementation of determiners in frames is more about involving full-fledged quantifying entities in the model. The first occurrence of this kind of entities (reminding semantic networks) lies in \cite{KallmeyerRichter:14}. Whole nodes have determiner types and scope attributes toward other elements of the FS. One advantage of this system is the included translation into Dominance Graphs. However, no other semantics is given to these nodes, and it fails to distinct individuals and set of individuals. Moreover, it does not go beyond the sentence level.

That is why another system has to be developed onto Feature Structures to be able to account for second-order entities out of a predicate to model properly quantification in frames.

%%%%%
\section{Feature Structures with wrappings}

%%%
\subsection{Determiners and sets in frames}

The expected system has to be an extension of the current framework for frames \cite{KallmeyerOsswald:13}, leading to the following criteria. It should provide a formal AVL syntax and Feature Structure semantics, with similar properties (subsumption, homomorphisms, minimal model). The supply should be powerful enough to account for special characteristics of quantification. Especially, variable sharing and scope constraints need special attention. But it has to be simple enough to be still translatable into First-Order Logic. Translation into other existent modellings (without needing them) would also be relevant to show how this system would behave in comparison to known ones. Following the formal pipeline of \cite{KallmeyerOsswald:13}, the system should be derived from an LTAG parsing to form an underspecified sentential semantic representation. Finally, as an additional contribution, we would like to be able to state the place of this representation with respect to the part representing the knowledge of the world (called context, or instance here).

In the idea of \cite{BarwiseCooper:81}, determiners could be modelled by relations between sets. Yet, frames fail systematically to model correctly sets and conjunctive elements because there is no notion of grouping. That is why one should be introduced. Frame Types, introduced in \cite{BaloghOsswald:19}, suggest to take a subframe and transform it into a single type. We follow this idea of collapsing a whole FS to a single node.

My contribution is this extension that I call \textbf{wrappings}, which adds to Feature Structures a sufficient and simple enough tool to model quantification in frames.

%%%
\subsection{The formal system}

\subsubsection{Syntax}

In order to define the syntax of \textbf{extended AVL} (Attribute-Value Logic), let us begin with the signature declaration. As usual, a finite set of type $\TYPE$ ($t$,...) and a finite set of attributes $\ATTR$ (\feat{p},\feat{Q},...) are necessary. Relations $\REL=\bigcup_{m>1}\REL_m$ ($\rel{r}_m$,...) of different possible arities will be used to bridge the need of some non-functional relationships. The standard labels are called base-labels or node names $\NNAME$ ($b$,$b_0$,$b_1$,...). Together with node variables $\NVAR$ ($x$,$y$,...) they are referred to as node labels $\NLABEL=\NNAME\uplus\NVAR$. Wrappings may be seen as envelopes around nodes, but which can be also treated as nodes themselves. This is enabled by including the set of wrapping variables $\WVAR$ ($T$,$S$,...) in the set of general labels $\LABEL=\NLABEL\uplus\WVAR$ ($k$,$l$,...). We work with infinitely many labels, but in practice this set is finite. 

\begin{equation}
\label{eq:avl}
\begin{array}{lrl}
     \AVD:& \varphi,\psi ::=& \feat{P}:\varphi~|~\type{t}~|~k~|~\varphi\wedge\psi~|~\varphi\vee\psi~|~\top\\
     \AVF:& \chi,\xi ::=& k\cdot\varphi~|~k\cdot p\trieq l\cdot q~|~\rel{r}_m((k_i\cdot p_i)_{i< m}) ~|~T:||x\cdot\varphi||~|~\chi\wedge\xi~|~\neg \chi~|~\top
\end{array}
\end{equation}

\noindent with $t\in\TYPE$, $\feat{P}\in\ATTR$, $k,k_i,l\in\LABEL$, $p,p_i,q\in\ATTR^*$, $x\in\NVAR$ and $\rel{r}_m\in\REL_m$ a relation of arity $m>1$ (types can be viewed as relations of arity 1). We also define some abbreviations : $\feat{P}_1\feat{P}_2...\feat{P}_n:\varphi$ instead of $\feat{P}_1:\feat{P}_2:...\feat{P}_n:\varphi$ and $k$ instead of $k\cdot\varepsilon$, where $\varepsilon$ is the empty word of $\ATTR^*$. The construction $k$ in $\AVD$ is just syntactic sugar for $...\trieq k\cdot\varepsilon$. Negation in $\AVD$ is admissible, i.e. not necessary to add because it could be expressed with negation in $\AVF$.

Except for the addition of disjunction in $\AVD$, negation and $\AVF$ and relations (which were already allowed in \cite{KallmeyerOsswald:13}), the only new constructor is $T:||x\cdot\phi||$. Attribute-Value Descriptions $\varphi\in\AVD$ are still thought of as one-place predicates. That is why they are used to describe the content of wrappings, symbolized by vertical double bars. The node variable $x$ labels $\phi$ and the wrapping variable $T$ labels the wrapping so that it could be referred to elsewhere in the formula. Note that the hierarchy between the two sorts forbids embedding of wrappings. Even if the way $\AVD$ are specified seems only to be able to describe trees (up to the disjunction), using several times the same wrapping variable allows every kind of directed graph. It must be remarked that attributes and relations can hold between nodes or wrappings, and even cross wrappings boundaries, what will be used. Eq.~\eqref{eq:abstr} gives an example of abstract AVL formula dealing with this.

\begin{equation}
\label{eq:abstr}
T_1:||x\cdot(\feat{p}:(\feat{p}:x\vee\feat{p}:y))||~\wedge~b\cdot\feat{q}:T_3~\wedge~T_2:||y\cdot \feat{p}:t||~\wedge~\rel{r}(b,y)~\wedge~x\cdot\feat{p}\trieq y\cdot\varepsilon~\wedge~T_3:||z\cdot\feat{Q}:\top|| %~\wedge~\neg z\cdot\top
\end{equation}

\subsubsection{Semantics}

We adopt a denotational semantics based on labelled and tagged directed graphs, i.e. with nodes and edges. We define a \textbf{Feature Structure with Wrappings} (FSW) on the signature $(\TYPE,\ATTR,$ $\REL,\NNAME,\NVAR,$ $\WVAR)$ as $F=\ <\!V,\W,\I\!>$ where $V$ is a nonempty set (the nodes), $\W\subseteq\wp(V)$ (the wrappings, with $\W\cap V=\emptyset$) is a set of disjointed nonempty subsets of $V$ ($\forall W,X\in\W,~W\cap X=\emptyset$ and $\forall W\in\W,~W\neq\emptyset$), and $\I$ is the interpretation function :

\begin{equation}
\label{eq:interpr}
 \I : \TYPE\to\wp(V),~~~\ATTR\to[V\rightharpoonup V],~~~\REL_m\to \wp(V^m),~~~\NNAME\rightharpoonup V
\end{equation}
   
\noindent for all $m>1$, where $\wp(Z)$ is the powerset of a set $Z$ and $f:Z\rightharpoonup Z'$ stands for partial function. We write $f(z)\downarrow$ for $z\in Z$ to mean $z\in\dom f$, i.e. $f$ is defined on $z$. We also write $Z\hookrightarrow Z'$ to mean an injection from $Z$ to $Z'$.

The nodes belonging to some wrapping are named wrapped nodes $\cup\W\eqdef\bigcup_{W\in\W}W$ and the remaining ones make up the complement $\comp\eqdef V\setminus\cup\W$. Therefore $\WS=\W\cup\{\comp\}$, called the set of \textbf{w-sets} ($P$,$Q$,...) for convenience, is a partition of $V$. So we can define $[v]$ for any $v\in V$ as the unique w-set $W\in\WS$ such that $v\in W$.

The expectation of such a definition is the properties it yields. The least we can ask is that the subsumption relation is an order (up to isomorphism). That is why \cite{KallmeyerOsswald:13} introduced the \textbf{reachability constraint} stipulating that every node be attribute-accessible from a base-labelled node. Each base-label is uniquely referring to a node, given by the interpretation function. Intuitively, a base-labelled node denotes an entity outside of an isolated model, which can be for example co-referred by anaphora. On the contrary, the idea of Frame Types considers general concepts (like the predicate used to build a set), not particular individuals anchored in the current world. So base-labels are not appropriate to label nodes in wrappings. However, it is a bad idea not to label nodes, even in wrappings. As Kallmeyer and Osswald noticed, the lack of the reachability constraint prevents subsumption to be an order. Namely, equivalence classes are too big and the relationship to information flow is lost. See Fig.~\ref{fs:eq-class} in appendix \ref{ap:fig} for an illustrated example.
% figs

The solution is to use variables as node labels in the model itself. By preventing variables to be bound in the syntax (especially extended AVL does not contain any quantifier), they can be reused freely. The idea behind that is that we may consider a model up to variable renaming, otherwise this choice would make no difference with base-labels. To have access to variables and their references we let a \textbf{model} be $M=\ <\!F,g\!>$, where $F$ is a FSW and $g=\ <\!\gN,\gW\!>$ is an assignment function with $\gN:\NVAR\rightharpoonup V$ and $\gW:\WVAR\rightharpoonup \W$. The reachability constraint can be now stated (notation just below):

\ex. \label{ex:reach} Every node in a w-set $P$ must be attribute-accessible by a variable or base-labelled node belonging to the same w-set.\\
    i.e. $\forall v\in V,~\exists k\in\NLABEL,p\in\ATTR^*,~\I_g(k)\in[v]\wedge\I(k,p)=v$

Now we can introduce the semantics of extended AVL formulas. For legibility we first define some notation:

\begin{equation}
\label{eq:i-ext}
\begin{array}{rl}
    \I_g(n) =& \I(n)\\
    \I_g(x) =& \gN(x)\\
    \I_g(T) =& \gW(T)\\
    \I_g(k,p\feat{q}) =& \I(\feat{q})(\I_g(k,p)) \text{ and } \I_g(k,\varepsilon) = \I_g(k)
\end{array}
\end{equation}

The two satisfaction relations are defined for a Feature Structure with Wrappings $F=\ <\!V,\W,\I\!>$, an assignment function $g$, a w-set $P\in\WS$ and a node $v\in P$.

\begin{equation}
\label{eq:sem}
\begin{array}{ll}
F,g,P,v\models t& \text{iff } v\in\I(t)\\
F,g,P,v\models \feat{p}:\varphi& \text{iff } F,g,P,\I(\feat{p})(v)\models\varphi\\
F,g,P,v\models k& \text{iff } \I_g(k)=v\\
F,g\models k\cdot\varphi& \text{iff } F,g,\comp,\I_g(k)\models \varphi\\
F,g\models k\cdot p\trieq l\cdot p& \text{iff }  \I_g(k,p)=\I_g(l,q)\\
F,g\models \rel{r}_m((k_i\cdot p_i)_{i<m})& \text{iff }  (\I_g(k_i,p_i))_{i<m}\in\I(\rel{r}_m)\\
F,g\models T:||x\cdot\varphi||& \text{iff } F,g,\I_g(T),\I_g(x)\models\varphi\\
\end{array}
\end{equation}

The boolean connectives are interpreted as usual. Remark that the constructor $k\cdot\varphi$ is the complement of $T:||x\cdot\varphi||$ in the sense that it requires the label to label a node in the complement set. The constructor $\trieq$, yet, does not require this and is a general statement about path equality. Moreover, some a priori authorized terms like $T\trieq b$ automatically leads to semantic failure thanks to $\W\cap V=\emptyset$. Fig.~\ref{fs:abstr} depicts the two minimal models of formula \eqref{eq:abstr}.
% in fact wrappings can have outgoing edges, but hard to write:
% x\cdot\top \wedge T\cdot p\trieq x

\begin{figure}[ht]
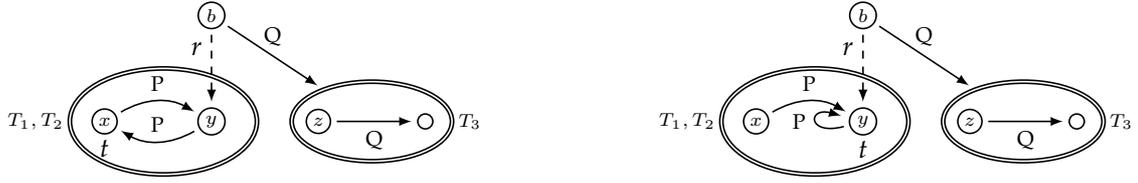

\begin{minipage}{0.5\textwidth}
\begin{fs}
\draw (22pt, 0pt)node[ft,minimum width=70pt,minimum height=40pt,label={180:\fsbase{$T_1,T_2$}}](ft1){};
\draw (0pt, 0pt)node[nod, label={-90:\type{t}}](0){\fsbase{$x$}};
\draw (40pt, 0pt)node[nod](1){\fsbase{$y$}};
\draw (40pt, 40pt)node[nod](2){\fsbase{$b$}};
\draw (100pt, 0pt)node[ft,minimum width=60pt,minimum height=30pt,label={0:\fsbase{$T_3$}}](ft3){};
\draw (80pt, 0pt)node[nod](4){\fsbase{$z$}};
\draw (120pt, 0pt)node[nod](5){};
\path[->] (0) edge[bend left] node[above]{\feat{p}}(1);
\path[->] (1) edge[bend left] node[above]{\feat{p}}(0);
\path[->] (2) edge[dashed] node[left,pos=0.2]{\rel{r}}(1);
\path[->] (2) edge node[above]{\feat{q}}(ft3);
\path[->] (4) edge node[below]{\feat{q}}(5);
% \draw[->] (1) .. controls (70pt,-10pt) and (60pt,-20pt) .. node[right]{\feat{q}}(ft);
\end{fs}
\end{minipage}
\begin{minipage}{0.5\textwidth}
\begin{fs}
\draw (22pt, 0pt)node[ft,minimum width=70pt,minimum height=40pt,label={180:\fsbase{$T_1,T_2$}}](ft1){};
\draw (0pt, 0pt)node[nod](0){\fsbase{$x$}};
\draw (40pt, 0pt)node[nod, label={-90:\type{t}}](1){\fsbase{$y$}};
\draw (40pt, 40pt)node[nod](2){\fsbase{$b$}};
\draw (100pt, 0pt)node[ft,minimum width=60pt,minimum height=30pt,label={0:\fsbase{$T_3$}}](ft3){};
\draw (80pt, 0pt)node[nod](4){\fsbase{$z$}};
\draw (120pt, 0pt)node[nod](5){};
\path[->] (0) edge[bend left] node[above]{\feat{p}}(1);
\path[->] (1) edge[loop left] node[left]{\feat{p}}(1);
\path[->] (2) edge[dashed] node[left,pos=0.2]{\rel{r}}(1);
\path[->] (2) edge node[above]{\feat{q}}(ft3);
\path[->] (4) edge node[below]{\feat{q}}(5);
% \draw[->] (1) .. controls (70pt,-10pt) and (60pt,-20pt) .. node[right]{\feat{q}}(ft);
\end{fs}
\end{minipage}
\caption{The two minimal models for formula \eqref{eq:abstr}}
\label{fs:abstr}
\end{figure}

One arising issue by leaving the possibility for attributes to cross wrapping boundaries is that it would require every node expected in a wrapping to be explicitly specified to belong to it (or to tell when an attribute crosses a wrapping boundary). To avoid that we add the \textbf{non-escapability constraint}:

\ex. \label{ex:non-esc} Attributes cannot cross wrapping boundaries from the inside to the outside\\
    i.e. Every entity (node or wrapping) attribute-accessible from a node in some wrapping is in the same wrapping\\
    i.e. $\forall W\in\W,v\in W,p\in\ATTR^*,~\I(p)(v)\downarrow~\to~\I(p)(v)\in[v]$
    
\noindent where $\I(p)(v)$ is defined as the canonical extension of $\I$ to attribute paths. Remark that it prevents a wrapping to be attribute-reachable from a wrapped node.

The apparent second-order capacity of wrappings is just an appearance. This comes from the fact that we cannot talk about any set of nodes, but only the ones present in the model and that they are disjoint. So a translation to First-Order Logic is possible.

\subsubsection{FOL translation}

Let us begin with the translation of the signature. The core was already considered in \cite{KallmeyerOsswald:13}: types, attributes and relations of arity $m$ are mapped to predicates of arity $1$, $2$ and $m$ respectively. Base-labels could be mapped to constant symbols, but this would require each of them to be present in any model. So, we better map them to one-place predicates. This is the same for node and wrapping variables, as they are in fact just used as labels. To avoid confusion, the FOL variables will be noted $u,w,...\in\X$. The translation of Attribute-Value Descriptions is processed while keeping the variable $u$ of evaluation.

\begin{equation}
\label{eq:trans-avd}
\begin{array}{rclcrclcrcl}
\overline{\feat{p}:\varphi}[u]&\rightsquigarrow&\exists w.~\feat{p}(u,w)\wedge\overline{\varphi}[w]&
& \overline{t}[u]&\rightsquigarrow& t(u)&
& \overline{k}[u]&\rightsquigarrow& k(u)
\end{array}
\end{equation}

\noindent and the boolean connectives just let the variable pass through.

As wrappings are full-fledged citizens of the frame model, they can be mapped to simple nodes with a special predicate $\WR$. Belonging to a wrapping is, without any difficulty, expressed by another additional $\IN(u,w)$ predicate (node $u$ is in the wrapping $w$). Hence the translation of Attribute-Value Formulas:

\begin{equation}
\label{eq:trans-avf}
\begin{array}{rcl}
\overline{k\cdot\varphi}&\rightsquigarrow& \exists u.~\neg\WR(u)\wedge(\forall w.~ \neg\IN(u,w))\wedge k(u)\wedge\overline{\varphi}[u]\\
\overline{k\cdot p\trieq l\cdot q}&\rightsquigarrow& \exists u,u',w.~k(u)\wedge l(u')\wedge p(u,w)\wedge q(u',w)\\
\overline{\rel{r}_m((k_i\cdot p_i)_{i<m})}&\rightsquigarrow& \exists (u)_{i<m},(w)_{i<m}.~(\bigwedge_{i<m} k_i(u_i)\wedge p_i(u_i,w_i))\wedge \rel{r}_m((w_i)_{i<m})\\
\overline{T:||x\cdot\varphi||}&\rightsquigarrow& \exists u,w.~\IN(u,w)\wedge T(w)\wedge x(u)\wedge\overline{\varphi}[u]\\
\end{array}
\end{equation}

\noindent where $\feat{p}_0...\feat{p}_{m-1}(u_0,u_m)$ is an abbreviation for $\exists (u_i)_{0<i<m}.~ \bigwedge_{0\leq i< m}\feat{p}_i(u_i,u_{i+1})$.

Of course, this is not sufficient to hope to deal only with frames. A lot of constraints still have to be expressed, for example, the functionality of attributes or the way $\IN$ implies $\WR$ for its second argument. This is taken into account by working in the theory of frames with wrappings satisfying the axiom schemata in \eqref{eq:theor}.

\begin{equation}
\label{eq:theor}
\begin{array}{lr}
\bigwedge_{\feat{p}}\forall u,w,w'.~\feat{p}(u,w)\wedge\feat{p}(u,w')~\to~w=w' &\text{(attribute functionality)}\\
\bigwedge_{k}\forall u,u'.~k(u,)\wedge k(u')~\to~u=u' &\text{(labelling uniqueness)}\\[2ex]
\forall u,w.~\IN(u,w)\wedge \WR(u)~\to~\bot &\text{(wrapping non-embeddedness)}\\
\forall w.~\WR(w)~\to~\exists u.~\IN(u,w) &\text{(wrapping non-emptiness)}\\
\forall u,w',w.~\IN(u,w)\wedge \IN(u,w')~\to~w=w' &\text{(wrapping disjointedness)}\\
\forall u,w.~\IN(u,w)~\to~\WR(w) &\text{($\WR$ and $\IN$ connection)}\\[2ex]
\forall u,w.~\bigvee_k\bigvee_p\exists u'.~k(u')\wedge p(u,u')\wedge (\IN(u,w)\to\IN(u',w)) &\text{(reachability constraint)}\\
\bigwedge_{\feat{p}}\forall u,u',w.~(\IN(u,w)\wedge\feat{p}(u,u'))~\to~\IN(u',w) &\text{(non-escapability constraint)}\\
\end{array}
\end{equation}

One corollary is the impossibility for wrapped nodes to point to wrappings with an attribute. The formulation of the property and its proof with only two of the above axioms is left to the reader.

The translation of models is even more straightforward and ensues directly from what was discussed above. Fig.~\ref{fs:abstr-fol} illustrates the translation of some model in FOL. The translation is correct, i.e. for every model $M$ and $\AVF$ $\chi$, $M\models\chi$ iff $\overline{M}\models\overline{\chi}$, but proof would be too long to be enclosed in this report. This could give us an upper bound for the model-checking complexity. Nevertheless, we will see that we can do much better thanks to the minimal model property.

\begin{figure}[ht]
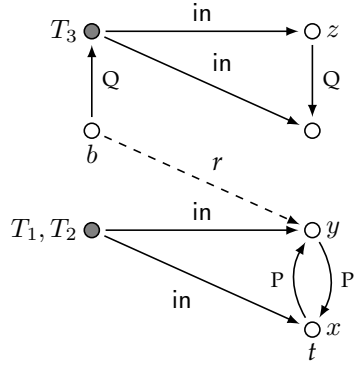

\begin{fs}
\draw (80.0, 40.0)node[nod, label={-90:$b$}](4){};
\draw (80.0, 0.0)node[nod,fill=gray,label={180:$T_1,T_2$}](ft1){};
\draw (80.0, 80.0)node[nod,fill=gray,label={180:$T_3$}](ft2){};
\draw (160.0, -40.0)node[nod, label={0:$x$}](0){};
\draw (160.0, -40.0)node[circle,label={-90:\type{t}}]{};
\draw (160.0, 0.0)node[nod,label={0:$y$}](1){};
\draw (160.0, 80.0)node[nod,label={0:$z$}](3){};
\draw (160.0, 40.0)node[nod](5){};
\path[->] (ft1) edge node[below left]{$\IN$}(0);
\path[->] (ft1) edge node[above]{$\IN$}(1);
\path[->] (0) edge[bend left] node[left]{$\feat{p}$}(1);
\path[->] (1) edge[bend left] node[right]{$\feat{p}$}(0);
\path[->] (ft2) edge node[above]{$\IN$}(3);
\path[->] (4) edge node[right]{$\feat{q}$}(ft2);
\path[->] (4) edge[dashed] node[above right]{\rel{r}}(1);
\path[->] (3) edge node[right]{$\feat{q}$}(5);
\path[->] (ft2) edge node[above right]{$\IN$}(5);
\end{fs}
\caption{Translation of one minimal model of \eqref{eq:abstr} (the one on the left of Fig.~\ref{fs:abstr}) in a FOL model. The gray nodes are the ones satisfying $\WR$.}
\label{fs:abstr-fol}
\end{figure}

\subsubsection{Homomorphism}

Extending subsumption to FSW is not involved. At first glance, it suffices to define homomorphisms with an additional mapping between wrappings, with the intuitive property that it conserves wrapping belonging, i.e. nodes belonging to some wrapping have to belong to the image wrapping. But for reasons we will explain later, considering the ability to transform a node into a wrapping is also interesting, so our system must account for that.

We say that a model $M$ subsumes a model $M'$ on the same signature, written $M\sqsubseteq M'$, if there exists a homomorphism $h=\ <\!\hV,\hW\!>$ where $\hV:V\to V'\uplus\W'$ and $\hW:\W\to\W'$ such that:

\begin{equation}
\label{eq:hom}
\begin{array}{lr}
\hV(\I(t)) \subseteq \I'(t) & \text{(types)}\\
\text{If } \I(\feat{p})(v)\downarrow \text{ then } \hV(\I(\feat{p})(v))=\I'(\feat{p})(\hV(v)) & \text{(attributes)}\\
\hV(\I(\rel{r}_m))\subseteq\I'(\rel{r}_m) & \text{(relations)}\\
\text{If } \I(b)\downarrow \text{ then } \hV(\I(b))=\I'(b)& \text{(base-labels)}\\
\text{If } v\in W \text{ then } \hV(v)\in \hW(W)& \text{(wrappings)}\\
\text{If } g(x)\downarrow \text{ then } \hV (g(x))=g'(x)& \text{(assignments)}
\end{array}
\end{equation}

\noindent for every $v\in V$, $t\in\TYPE$, $\feat{p}\in\ATTR$, $\rel{r}_m\in\REL_m$, $b\in\NNAME$ and $W\in\W$. By abuse of notation we can simply use $h$ instead of $\hV$ or $\hW$.

A model $M'$ is isomorphic to $M$ if there exists a homomorphism $\sigma:M\to M'$ (called isomorphism) such that $\codom~ \sigma_V\subseteq V'$, being bijective and such that $\sigma^{-1}$ is a homomorphism. Even with such a definition, we keep the useful property \ref{th:order} (in appendix \ref{ap:proof}) that $\sqsubseteq$ induces an ordering up to isomorphism. This is partly due to the property \ref{th:unique} stating that there is at most one homomorphism between two models. Moreover, we can still define unification as a least upper bound (lub), noted $M\sqcup M'$ when it exists.

The potential variable renaming might be taken into account. We say that $M$ $\alpha$-subsumes $M'$, written $M\sqsubseteq_{\alpha}M'$ if there is a homomorphism $h:M\to M'$ (called $\alpha$-homomorphism) having moreover a submapping $\hvar:\dom~\gN\uplus\dom~\gW\hookrightarrow\dom~\gN'\uplus\dom~\gW'$ verifying the property \eqref{eq:hom-alpha} instead of (assignments) of \eqref{eq:hom}.

\begin{equation}
\label{eq:hom-alpha}
\begin{array}{lr}
\text{If } g(x)\downarrow \text{ then } \hV(g(x))=g'(\hvar(x))& ~~~~~~~~~~ \text{(}\alpha\text{-assignments)}
\end{array}
\end{equation}

\noindent with $x\in\NVAR\uplus\WVAR$. Remark that the definition of $\hvar$ is just a trick to state concisely the injectivity on all the domain set. In fact no wrapping variable can be mapped into a node variable, otherwise, it would contradict the definition of $\hW$.

%example?

A model $M'$ is \boldmath$\alpha$\unboldmath\textbf{-equivalent} to $M$ if there exists some $\alpha$-homomorphism $\sigma:M\to M'$ (called $\alpha$-isomorphism) with $\codom~ \sigma_V\subseteq V'$, $\svar(\dom~\gN)\subseteq \NVAR$, being bijective and such that $\sigma^{-1}$ is an $\alpha$-homomorphism. Isomorphism and $\alpha$-equivalence are clearly equivalence relations. The great property is that mutual subsumption ($\alpha$-subsumption resp.) is equivalent to isomorphism ($\alpha$-equivalence resp.). Therefore we can extend it to define subsumption between equivalence classes, especially $\alpha$-equivalence classes. If $h:M_1\to M_2$ is an $\alpha$-homomorphism and $\sigma_1:M_1\to M_1'$ and $\sigma_2:M_2\to M_2'$ are $\alpha$-isomorphisms, then by composing we can create an $\alpha$-homomorphism $h'\eqdef\sigma_2\circ h\circ \sigma_1^{-1}:M_1'\to M_2'$. This allows us to claim that the $\alpha$-equivalence class $[M_1]_{\alpha}$ of $M_1$ $\alpha$-subsumes that of $M_2$: $[M_1]_{\alpha}\sqsubseteq_{\alpha}[M_2]_{\alpha}$.

The interesting consequence is that $\alpha$-equivalence classes also form an ordering for subsumption. However, variable renaming prevents some couple of models to have any least upper bound, contrary to the intuition. This raises a big issue because this was the way we used to define unification. Fig.~\ref{fs:lub} in appendix \ref{ap:fig} shows some conflicting example. Moreover, it would in fact also be a strange idea to consider the $\models$ relation up to variable renaming. This system is a bit too powerful and does not generate the wanted behaviour. Consequently $\alpha$-homomorphisms will not be considered any longer. The use of variables as labels therefore only brings a conceptual insight.

%%%
\subsection{Possible theories and interpretations}

\subsubsection{Representation of determiners}

Now that we possess a logical tool able to group nodes, a theory of quantifiers in frames can be built. Several approaches are possible. Quantifiers are indeed multifaceted. They have logical meanings but which are very sensitive to the syntax (and the context). We do not have the ambition to capture all the subtle meanings generated by quantifiers. Nevertheless, two distinct aspects goals compete here.

The first incentive is to provide a cognitive-adequate representation of determiners. However, no work on the mental representation of non-numerical quantification is known by us. A lot of semantic models propose explanations for precise lexical units, but no consensus has been met. In spite of this, there are heuristics that we can follow. Quantification is for sure not considered an individual or an event, namely not as strongly ``pictorial'' as this kind of entities. So using nodes to represent it is dubious.

That is why a first proposal is to utilize a relation. By trying to stick to previous attempts, let us construe this relation as a relationship between two sets, the restrictor and the nuclear scope. As stated before both may be represented by wrappings, which interpretation is the set of actual realizations (called instances) of their contents. But seeing sets as one-place predicates, some solution has to be found to express that quantified variable. Actual variable sharing cannot function. Indeed a variable of $\NVAR$ labels only one node, so it would be impossible to make two different wrappings with both having a node labelled by the same variable. But as wrapping boundaries are a bit permeable, this may be used to connect nodes of different wrappings.

In Fig.~\ref{fs:bark-dog} on the left, the depiction of the frame for \ref{ex:bark-dog} in that theory highlights the connection between the \type{dog} node and the \feat{agent} of the \type{barking} node. This means that they represent the same individual in the quantification. Stated differently, \textit{Every dog is a barker}. Formally in AVL, if $\chi$ is a set of $\AVD$ (or directly a fact-set, see section \ref{sub:fact-set}) describing the restrictor wrapping, $\xi$ for the nuclear scope wrapping, and $\rel{every}(k\cdot p,l\cdot q)$ holds, then the associated formula is
\begin{equation}
\label{eq:wrap-interpr}
\forall u.~(\exists \bar{x}.~(\chi\wedge k\cdot p\trieq u))~\to~ (\exists \bar{y}.~(\xi\wedge l\cdot q\trieq u))
\end{equation}
where $\bar{x}$ ($\bar{y}$ resp.) is the set of free variables in $\chi$ (in $\xi$ resp.) and $u$ is a fresh variable. This should be read ``Every individual enjoying the property $\chi$ is an individual enjoying the property $\xi$''. 
This might give another support for the use of variables as labels in wrappings. It is also closely related to how we might take advantage of the structure for the discourse level (see section \ref{sub:inst}).

\begin{figure}[ht]
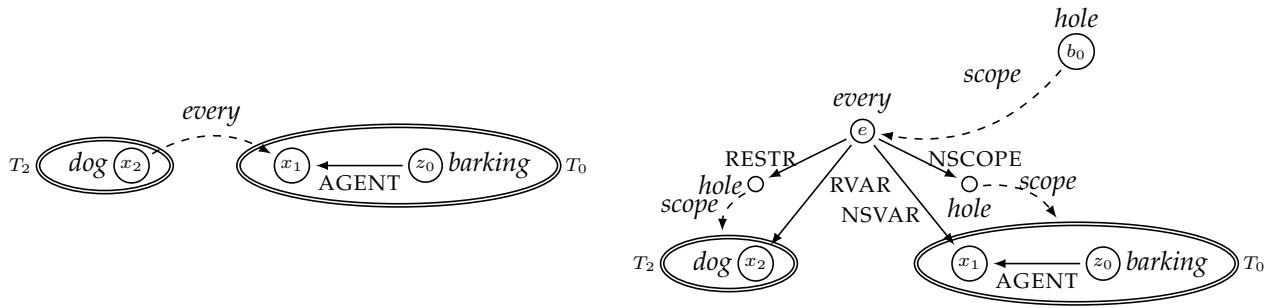

\begin{minipage}{0.5\textwidth}
\begin{fs}
\draw (-10pt, 0pt)node[ft,minimum width=50pt,minimum height=20pt,label={180:\fsbase{$T_2$}}](ft1){};
\begin{scope}
    \draw (0pt, 0pt)node[nod,label={180:\type{dog}}](11){\fsbase{$x_2$}};
\end{scope};
\draw (100pt, 0pt)node[ft,minimum width=120pt,minimum height=30pt,label={0:\fsbase{$T_0$}}](ft0){};
\begin{scope}[xshift=60]
    \draw (0pt, 0pt)node[nod](01){\fsbase{$x_1$}};
    \draw (50pt, 0pt)node[nod,label={0:\type{barking}}](00){\fsbase{$z_0$}};
    \path[->] (00) edge node[below]{\feat{agent}}(01);
\end{scope};
\path[->] (11) edge[bend left,dashed] node[above]{$\rel{every}$}(01);
\end{fs}
\end{minipage}
\begin{minipage}{0.5\textwidth}
\begin{fs}
\draw (-10pt, 0pt)node[ft,minimum width=50pt,minimum height=20pt,label={180:\fsbase{$T_2$}}](ft1){};
\begin{scope}
    \draw (0pt, 0pt)node[nod, label={180:\type{dog}}](11){\fsbase{$x_2$}};
\end{scope};

\draw (120pt, 0pt)node[ft,minimum width=120pt,minimum height=30pt,label={0:\fsbase{$T_0$}}](ft0){};
\begin{scope}[xshift=80]
    \draw (0pt, 0pt)node[nod](01){\fsbase{$x_1$}};
    \draw (50pt, 0pt)node[nod,label={0:\type{barking}}](00){\fsbase{$z_0$}};
    \path[->] (00) edge node[below]{\feat{agent}}(01);
\end{scope};

\draw (40pt, 50pt)node[nod,label={90:\type{every}}](ev){\fsbase{$e$}};
\draw (0pt, 30pt)node[nod,label={180:\type{hole}}](evr){};
\draw (80pt, 30pt)node[nod,label={-90:\type{hole}}](evns){};
\path[->] (ev) edge node[left]{$\feat{restr}$}(evr);
\path[->] (evr) edge[bend right,dashed] node[left]{$\type{scope}$}(ft1);
\path[->] (ev) edge node[right,pos=0.4]{$\feat{rvar}$}(11);
\path[->] (ev) edge node[right]{$\feat{nscope}$}(evns);
\path[->] (evns) edge[bend left,dashed] node[right,pos=0.3]{$\type{scope}$}(ft0);
\path[->] (ev) edge node[left,pos=0.7]{$\feat{nsvar}$}(01);
\draw (120pt, 80pt)node[nod,label={90:\type{hole}}](h){\fsbase{$b_0$}};
\path[->] (h) edge[bend left,dashed] node[above left,pos=0.2]{$\type{scope}$}(ev);
\end{fs}
\end{minipage}
\caption{Two frames for \ref{ex:bark-dog}. On the left, following a cognitive-friendly approach. On the right, following a syntactic approach.}
\label{fs:bark-dog}
\end{figure}

This theory leads to underspecified representations in the sense that two determiner relations could point to different nodes of the same wrapping without specifying which one outscopes the other. However, a solved (i.e. unambiguous) form is not representable. This stems from the impossibility of relations to be entities which can themselves have ingoing edges (e.g. from another determiner). Consequently, granting more powerful entities to determiners is a good strategy.

The second theory, named \textbf{Quantified Complexes} takes this remark into account, and posits a more syntactic view of quantification. Close to type-logical formulas, a determiner incarnates a full-fledged subframe representing its term: the root node with its type, a restrictor node reachable by \feat{restr} and a nuclear scope node reachable by \feat{nscope}. Reminiscent of Hole Semantics, the latter are so-called holes. This means that they can receive another subformula (or here subframe). To mimic Dominance Graphs dominance constraints are introduced as relations \rel{scope} from a hole to a determiner (or more generally logical) node or wrapping. Thus wrappings obtain an interpretation of subformulas and make possible the cut of subframes into formulas\footnote{This is especially important to derive correct every unambiguous reading}. But as soon as the scope points to the entire wrapping predicate, variable sharing has to be spotted another way. This is simply done with two additional attributes \feat{rvar} (restrictor variable) and \feat{nsvar} (nuclear scope variable). All wrappings are stuck together by logical subframes, hence the denomination of Quantified Complex. Fig.~\ref{fs:bark-dog} on the right displays an example for \ref{ex:bark-dog}.

\subsubsection{Interpretation of wrappings}

Albeit formally clear how quantifiers relate to predicates, nothing was claimed about the relationship between wrappings and current linguistic theories. There again, two distinct approaches make different predictions.

The first interpretation is the \textbf{Frame Types} (FT) of \cite{BaloghOsswald:19}. They introduce it to deduce automatically the type of a subevent of a global event. For example, a particular state (given by a point in time) in a scalar change (e.g. the process of becoming dryer) has to be an instance of the Frame Type that the scalar change is an instance of. Their formalism distances itself from ours because they use any $\AVD$ to be the name (i.e. the label) of FT nodes. This has the drawback to be unable to specify multirooted Feature Structures in wrappings. But as they suppose that the unique root is the center node, it makes sense to use only edges between FTs, as opposed to our \feat{rvar} and \feat{nsvar}. At least, our system is more expressive than theirs.

The main statement with a FT interpretation is that the content of a wrapping determines it, i.e. there cannot exist two wrappings of the same content. The idea is to be in the capacity of putting some type hierarchy in the model, without specifying it explicitly. The \rel{sub} relation between relative centers tells that the individual of the first FT has as subtype the individual of the second FT. To me, this can be extended as a particular case of a subsumption relation (in the other direction). This could also include the type hierarchy itself. An illustration can be found in Fig.~\ref{fs:ft}.

\begin{figure}[ht]
\begin{fs}
\draw (-40pt, 65pt)node[ft,minimum width=70pt,minimum height=50pt](ft3){};
\begin{scope}[xshift=-20,yshift=60]
    \draw (0pt, 0pt)node[nod,label={90:\type{dog}}](31){};
    \draw (-40pt, 0pt)node[nod,label={90:\type{big}}](32){};
    \path[->] (31) edge node[below]{\feat{size}}(32);
\end{scope};
\draw (10pt, 0pt)node[ft,minimum width=40pt,minimum height=20pt](ft1){};
\begin{scope}
    \draw (0pt, 0pt)node[nod,label={0:\type{dog}}](11){};
\end{scope};
\draw (120pt, 0pt)node[ft,minimum width=120pt,minimum height=30pt](ft0){};
\begin{scope}[xshift=80]
    \draw (0pt, 0pt)node[nod](01){};
    \draw (50pt, 0pt)node[nod,label={0:\type{barking}}](00){};
    \path[->] (00) edge node[below]{\feat{agent}}(01);
\end{scope};
\draw (75pt, 60pt)node[ft,minimum width=140pt,minimum height=30pt](ft2){};
\begin{scope}[xshift=40,yshift=60]
    \draw (0pt, 0pt)node[nod,label={180:\type{dog}}](21){};
    \draw (50pt, 0pt)node[nod,label={0:\type{barking}}](20){};
    \path[->] (20) edge node[below]{\feat{agent}}(21);
\end{scope};
\draw (60pt,100pt)node(p){...};
\draw (-100pt, -20pt)node[nod,label={-30:\type{dog}}](1){\fsbase{$b$}};
\draw (-140pt, 0pt)node[nod,label={180:\type{big}}](2){};
\draw (-140pt, -40pt)node[nod,label={180:\type{Bull}}](3){};
\path[->] (1) edge node[below left]{\feat{size}}(2);
\path[->] (1) edge node[below right]{\feat{name}}(3);
\path[<-] (00) edge[dashed] node[right]{$\rel{sub}$}(20);
\path[<-] (01) edge[dashed] node[right]{$\rel{sub}$}(21);
\path[<-] (11) edge[dashed] node[right]{$\rel{sub}$}(21);
\path[<-] (11) edge[dashed] node[left]{$\rel{sub}$}(31);
\path[<-] (20) edge[dashed] node[right,pos=0.7]{$\rel{sub}$}(p);
\path[->] (1) edge[dashed] node[above left]{$\rel{inst}$}(31);
\path[->] (1) edge[dashed] node[below right]{$\rel{inst}$}(11);
\path[->] (2) edge[dashed] node[above left]{$\rel{inst}$}(32);
\end{fs}
\caption{Illustration of the Frame Type theory. In the upper-right corner, some FTs are shown with their respective \rel{sub} relations. In the bottom-left corner, a part of an instance (here the knowledge that there exists a big dog named Bull) is connected with its corresponding FTs. Labels are not represented, except in the instance where base-labels spot individuals of the outside world.}
\label{fs:ft}
\end{figure}

The \rel{inst} relation tells that the first argument is an instance of the second (which should be a wrapped node). The instance, unlike the representation of a sentence, is the part of a frame describing the world knowledge, the current context. With this idea of \rel{inst} relation the point was to update automatically the instance, knowing of what FT a node was, but also maybe conversely to guess the FT of a subframe; the whole using unordered local AVL constraints (see section \ref{subsub:constraints}). While this looked possible on simple examples, this system fails to capture the complexity of quantified sentences, especially because that phenomenon is not local at all and very sensitive to the overall environment. Another method is discussed in section \ref{sub:inst}.

Furthermore, FTs figure out to be useful to declare a general kind of sets. By enabling some nodes to have a \type{collection} type, an attribute \feat{ctype} (collection type) could go to a FT to specify the kind of its elements by an intentional definition. This would also allow extensional definitions of sets. This could solve the recurring problem of frames to fail to represent correctly conjunction.

Nonetheless, representation issues might arise at the discourse level. If Quantified Complexes are put up on FTs, non-related occurrences of the same predicate would make their common FT receive two non-related quantifier attributes. The syntactic theory could differentiate this by viewing the main hole $b_0$ scoping over every determiner involved in a sentence as an anchor of the Quantified Complex for that sentence (so, differentiating sentences). But this does not seem satisfactory and is not elegant anyway, hence an alternative interpretation.

The second interpretation of wrappings is called \textbf{Uninstantiated Frames} (UF). It posits less internal structure (no \rel{sub} relation) and authorizes different wrappings of the same content. The idea is just that wrappings are subframes able to be instantiated. This better connects to a syntactic representation, where wrappings can be just saw as a (set of) formula(s) describing a predicate. Same wrappings might be still employed in order to account for anaphoric phenomena.

%%%%%
\section{From the lexical to the discourse representation}

%%%
\subsection{Lexical minimal models}
\label{sub:fact-set}

In an implementation of frames, the lexicon would be a list of couples of an anchored elementary tree and a semantic model. But for compactness and convenience, it would be rather more interesting to specify the latter with AVL formulas. So, a minimal model must be computed to finalize the lexicon constitution. Our goal is to extend the work of \cite{Hegner:93} of minimal FS model to models with variables and wrappings. The paper only considered Horn formulas (i.e. no positive disjunction) but proved that this is the larger class able to be treated in polynomial time. However, disjunctions appear to be relevant in frames, may it be to specify faster several meanings of a same elementary tree, for properties which can be distinguished only at the semantic sentence formation, or for words which might be always inherently ambiguous (like dot objects \cite{BabonnaudKallmeyerOsswald:16}). So, expanding the tool to accept several minimal models is an expected milestone.

To begin with, let us extend the $\models$ notion. Conventionally, an $\AVF$ $\chi$ entails $\xi$, written $\chi\models\xi$, if for every model $M$, $M\models\chi$ implies $M\models\xi$. They are said to be equivalent if the relation holds in both direction, abbreviated by $\chi\iseq\xi$. Similarly, we define $\chi$ to entail a model $M$ if for every model $M'$, $M'\models\chi$ implies $M\sqsubseteq M'$ (what is also written $M'\models M$). So $M$ is a minimal model for $\chi$ if and only if $M\iseq\chi$.

The extension to model sets follows the intuition. Set $\M$ a finite set of models. In the following, only finite sets are considered. Define $\M\models\chi$ by: for every $M\in\M$, $M\models\chi$. The reciprocal definition should draw more attention: $\chi\models\M$ iff for every model $M'$, there exists $M\in\M$ such that $M'\models M$. A \textbf{Minimal Model Set} (MMS) $\M$ is therefore as model set verifying $\M\iseq\chi$. We say that it is \textbf{canonical} if moreover, it is an antichain for $\sqsubseteq$. Canonical Minimal Model Sets are the best to represent formulas because they are unique, as proven in theorem \ref{th:cmms} in appendix \ref{ap:proof}. Note that $\models$ is not ``completely'' transitive because of negation in $\AVF$.

The first step of the Minimal Model Set construction is transforming $\chi$ into a Conjunctive Normal Form (CNF) $\xi$. To do so, we have to define atoms, that have a slightly different syntax and semantics.

\begin{equation}
\label{eq:atom}
\ATOMS::=k\bl p:\top~|~k\bl p:t~|~k\bl p\trieq l\bl p~|~\rel{r}_m((k_i\bl p_i)_{i<m})~|~\INWR(k)~|~\IN(k,T)~|~\WR(k)
\end{equation}

\begin{equation}
\label{eq:atom-sem}
\begin{array}{ll}
F,g\models k\bl p:\top& \text{iff } \I_g(k,p)\downarrow\\
F,g\models k\bl p:t& \text{iff }  \I_g(k,p)\in\I(t)\\
F,g\models \INWR(k)& \text{iff }  \I_g(k)\in V\wedge [k]\in\W\\
F,g\models \IN(k,T)& \text{iff } \I_g(k)\in\I_g(T)\\
F,g\models \WR(k)& \text{iff } \I_g(k)\in\W\\
\end{array}
\end{equation}

The other symbols are interpreted the same. This flat semantics allows us to spot clearly what are the facts holding (or not) in $\chi$. The Conjunctive Normal Form $\xi$ can be viewed as a set of sets $C$ (called clauses) of the form $C=\Sigma\Rightarrow \Delta$, with $\Sigma$ interpreted conjunctively and $\Delta$ disjunctively ($C=(\bigwedge_{\kappa\in\Sigma}\kappa) \Rightarrow (\bigvee_{\rho\in \Delta}\rho))$. It is in Horn form iff $|\Delta|\leq 1$. To connect to $\models$, we can also write $\xi=\bigwedge_{C\in\xi}C$. The translation into an equivalent CNF is not explained in details here, given that it is common practise. Moreover, the rewriting rules are very intuitive: the boolean connectives are pulled outside $\AVD$, the negations pushed to the atoms, and the syntactic sugar $k$ is replaced by $\trieq$. We only give the rules for the two main constructors:

\begin{equation}
\label{eq:rewriting}
k\cdot\varphi~\rightsquigarrow \neg\INWR(k)\wedge\neg\WR(k)\wedge k\bl\varphi~~~~~~~~~~~~~~
T:||x\cdot\varphi||~\rightsquigarrow \IN(x,T)\wedge x\bl\varphi
\end{equation}

Translation of \eqref{eq:abstr} is given in \eqref{eq:abstr-cnf}.

\begin{equation}
\label{eq:abstr-cnf}
\begin{array}{c}
\IN(x,T_1)\wedge(x\bl\feat{pp}\trieq x\bl\varepsilon\vee x\bl\feat{pp}\trieq y\bl\varepsilon)\wedge \neg\INWR(b)\wedge\neg\WR(b)\wedge b\bl\feat{q}\trieq T_3\bl\varepsilon\\
\wedge\IN(y,T_2)\wedge y\bl\feat{p}:t\wedge\rel{r}(b,y)\wedge x\bl\feat{p}\trieq y\bl\varepsilon\wedge\IN(z,T_3)\wedge z\bl\feat{q}:\top
\end{array}
\end{equation}

The second step consists in creating a (set of) \textbf{fact-set}(s).
The CNF with this special clause shape allows us to use the technique used by Hegner for Horn constraints. This is just the general algorithm for HornSAT, but instead of keeping the list of true propositions, we keep a list of facts (i.e. atoms) that hold: the fact-set. But given disjunction, several fact-sets can be necessary to account for possibly several minimal models.

Formally the saturation algorithm (in appendix \ref{ap:satur}) builds an equivalent finite set $\U$ of finite fact-sets $U\subseteq\ATOMS$ which should be interpreted as $\U=\bigvee_{U\in\U}(\bigwedge_{\beta\in U}\beta)$. Each fact-set $U$ represents a primitive syntactic version of a model $M$. The contradictory fact-sets are discarded. Equation \eqref{eq:abstr-fact} in appendix \ref{ap:fig} displays a sufficient part of the fact-sets created from \eqref{eq:abstr-cnf}.

As cycles may occur, we follow the proposal of \cite{Hegner:93} to limit the fact-sets to atoms with a maximum path length. But as we do not know the number of nodes of the final model, we can only take an approximated upper bound $C_0\eqdef L(\xi)+\mea{\xi}$, where $\mea{\xi}$ is the number of symbols in $\xi$ that are attributes. If $M$ is the final model, $C_0\geq |V|$ because each label or attribute creates at most one new node. So limiting the number of attribute symbols in any atom by $C_0$ ensures a finite set, because the deduction procedure (rules in \eqref{eq:deduce}) only adds a finite number of atoms (by induction on the rules).

Moreover, each fact-set $U$ is \boldmath$C_0$\unboldmath\textbf{-saturated}, i.e. for all $U'\subseteq U$ and $U'\models\beta'$ with $\mea{\beta'}\leq C_0$, then $\beta'\in U$. This allows us to have the appealing property \ref{th:reduc} (appendix \ref{ap:proof}) stating that $U\models U'$ iff $U'\subseteq U$. Consequently, the third step is to make sure that $\U$ is an antichain for $\subseteq$, which gives us a canonical fact-set $\U_{\mathsf{can}}$. This could be done naively by checking every pair. A better algorithm surely exists, especially with the knowledge of how $\U$ was created, but this goes beyond the scope of this research.

The fourth and last step consists in constructing a model $M$ out of each fact-set $U$. The process is solved by the following high-level algorithm. This yields a canonical Minimal Model Set.

\begin{itemize}
    \item First, create a node for every label present in $U$. Then, create paths with new nodes, preserving the functionality of attributes. It means: to process $k\cdot p:\top$, explore $p$ from the node for $k$ and whenever a node lacks, create a new one with the correct attribute to it.
    \item Then process the $\trieq$ atoms by merging the respective nodes.
    \item Then add types and relations.
    \item Finally wrappings are computed with a set data structure, for example Tarjan's. To begin with, assign a set $\mathcal{T}$ of wrapping labels to every $\INWR$-typed node $v$ labelled by $L$ (a set of labels), such that $\mathcal{T}=\{T\in\WVAR~|~\exists k\in L,~\IN(k,T)\in U\}$. Then, explore the graph by putting every node $v$ reachable with an attribute path from $w$ in the same wrapping set as $w$ (non-escapability constraint). 
\end{itemize}

No conflict should occur thanks to saturation. Remark that, without any particular constraint, nodes which are attribute-accessible from nodes in the complement set $\comp$ (i.e. from non-wrapped and non-wrapping nodes) also belong to the complement set. Indeed, belonging to a wrapping would be strictly more informative w.r.t. subsumption and subsumption prevents wrapped nodes to be unwrapped.

Given that $\U_{\mathsf{can}}$ is canonical, so is $\M$. But the anti-chain processing could be also performed directly on models, especially because of property \ref{th:unique}.

The complexity of the whole process is clearly not polynomial. Conversion into CNF is exponential in the size of connectives in the worst case. The saturation process is exponential in the number of non-Horn clauses if we count the deduction as atomic. The complexity of the deduction has not been investigated yet, but it seems to be polynomial with a good data structure (e.g. a dictionary to be able to use the typing of deduction rules to find rapidly all triggerable rules). Finally, the model construction is polynomial in the size of the fact-set.

This gives us a faster way to model check $M$ on $\chi$, by constructing the minimal model set $\M$ of $\chi$ and checking if any $M'\in\M$ subsumes $M$. If there is no disjunction and no negation in $\chi$ this problem belongs to P.

%%%
\subsection{Syntax-Semantics Interface}

\subsubsection{Syntax derivation and unification}

Following the framework of \cite{KallmeyerOsswald:13} we use LTAGs with features to parse sentences. Let us only explain the derivations for the syntactic theory of Quantified Complexes, the semantic theory requires fewer details. The choice of elementary tree and adequate syntax-semantics features is the same as in \cite{KallmeyerRichter:14}. Only a few changes are made to account for the differences with that modelling.

Determiners are as usual modelled by an auxiliary tree adjoining on the NP node of an initial tree anchoring a noun (or noun phrase). This NP then substitutes to the NP node of a verb (or preposition), so as its subject or object. The feature \feat{i} represents the individual making the main contribution of the NP phrase. We only study cases where we can analyze it so (i.e. no conjunction \textit{and, but,...}), for simplicity. This individual corresponds to the restrictor (or nuclear scope) variable of the quantifier. Scope is given by the predicate expected from the verb (or preposition) argument. It is transmitted by the feature \feat{p}. On the verb, it stands for the considered action or property. On the noun, it stands for the type of the individual but also potential adjectival adjoints. Therefore they have to connect to wrappings. The maximum scope attribute \feat{m} is there to help the formation of dominance edges. It plays a critical role for more complex sentences, enabling fundamental scope distinction between relative clauses and other adjuncts (see Fig;~\ref{fs:donkey} in appendix \ref{ap:fig}). Fig.~\ref{tag:bark-dog} depicts the derivation of \ref{ex:bark-dog}, giving the final frame in Fig.~\ref{fs:bark-dog} on the right.

\begin{figure}[ht]
\begin{fs}
\begin{scope}[shift={(-100pt,-60pt)}]
\Tree 
[.\node(np2){NP$_{\text{[{\sc i}:$x_2$,{\sc p}:$T_2$]}}$}; [.N \type{dog} ] ]
\end{scope}
\begin{scope}[shift={(-100pt,-130pt)}]
\draw (0,0) node {$T_2:||x_2\cdot dog||$};
\end{scope}
\begin{scope}
\Tree 
[.S \node(np0){NP$^{\text{[{\sc i}:$x_1$,{\sc p}:$T_0$,{\sc m}:$b$]}}$}; [.VP [.V \textit{barks} ] ] ]
\end{scope}
\begin{scope}[shift={(20pt,-110pt)}]
\draw (0,0) node {\begin{tabular}{l}
$T_0:||z_0\cdot(barking\wedge\feat{agent}:x_1)||$\\
$\wedge b\cdot \type{hole}$
\end{tabular}};
\end{scope}
\begin{scope}[shift={(-200pt,-20pt)}]
\Tree 
[. \node(np3){NP$_{\text{[{\sc i}:$y_1$,{\sc p}:$S_0$,{\sc m}:$l$]}}$}; [.D $every$ ] [.NP*$^{\text{[{\sc i}:$y_2$,{\sc p}:$S_2$]}}$ ] ]
\end{scope}
\begin{scope}[shift={(-230pt,-110pt)}]
\draw (0,0) node {\begin{tabular}{ll}
$e\cdot($&$every\wedge$\\
&$\feat{rvar}:y_2\wedge\feat{nsvar}:y_1\wedge$\\
&$\feat{restr}:hole\wedge\feat{nscope}:hole~)$\\
$\wedge$&$\type{scope}(e\cdot\feat{restr},S_2)\wedge\type{scope}(l,e)$\\
$\wedge$&$\type{scope}(e\cdot\feat{nscope},S_0)$\end{tabular}};
\end{scope}
\draw[adj,bend right] (np2) to (np0);
\draw[adj,bend left] (np3) to (np2);
\end{fs}
\caption{\label{tag:bark-dog} Derivation of \ref{ex:bark-dog}}
\end{figure}

The derivation triggers FS Value unifications gathered in a set $E$ of equality atoms. If each elementary tree $i<m$ is associated with a Canonical Minimal Model Set $\M_i$, the Canonical Minimal Model Set of the sentence is the set $\M$ of the existing $M_E\sqcup\bigsquplus_{i<m}M_i$ for any $(M_i)_{i<m}\in\prod_{i<m}\M_i$, where $M_E$ is the model produced out of the fact-set $E$. Remark that as each word model has distinct variables, the unification between them just amounts to juxtaposition ($\squplus$). Each element of $\M$ is a possible meaning of the utterance. But to be able to sort out the semantic deviant meanings, we need to specify additional semantic constraints.

\subsubsection{Constraints}
\label{subsub:constraints}

\textbf{Constraints} are properties that force the model to have a systematic and felicitous semantic structure. Already introduced in \cite{KallmeyerOsswald:13} they take the form of universally quantified one-place AVL predicates. We use here the $\bl$ notation to refer to every node, no matter wrapped or not. The first use is to encode the type hierarchy, e.g. $\forall x.~x\bl \type{cat}\to x\bl\type{animal}$. This could be viewed as a particular case of cascades. The latter are ontological rules on subframes, e.g. $\forall x.~x\bl\type{activity}\to x\bl\feat{actor}:\top$ or $\forall x.~x\bl(\type{activity}\wedge\type{motion})\to x\bl\feat{actor}\trieq x\bl\feat{mover}$. Thanks to them it is possible not to specify the \feat{actor} of an \type{activity} when it is unknown, while still being sure it is present in the model. Indeed, it could be implicit, and therefore co-referred to later in the discourse, so it needs to be there. Moreover, it provides a better decomposition of frames, what helps to be more systematic and prevents redundancy. The second main supply of constraints is to specify deviant meanings by incompatibility rules, like $\forall x.~x\bl(\type{entity}\wedge\type{substance})\to\bot$(e.g. a \type{cat} cannot be \type{lava}).

Constraints are specified in Horn clauses. As long it does not entail node creation, they just have to be triggered at most once on every node. So incorporating these constraints to the unification procedure is tractable.

The Quantified Complex extension comprises the same kinds of constraints. The type hierarchy of logical nodes (mainly devoted for function words) in Fig.~\ref{hier} together with the incompatibility constraints $\forall x.~x\bl (\type{ext}\wedge \type{logical})\to\bot$, where \type{ext} (extension\footnote{The word ``extension'' is not used here as in model-theoretic semantics to indicate a set of individuals, but rather as a synonym for instance}) is the most general type of nodes in the instance, ensures that nodes of the instance and nodes used in Quantifier Complexes will not merge. The type \type{conj} is used for semantic conjunction (a case of connective \type{conn}) but which can lead to quantification in subphrases (e.g. relative clauses), like the complex example of Fig~\ref{fs:donkey}.

\begin{figure}[ht]
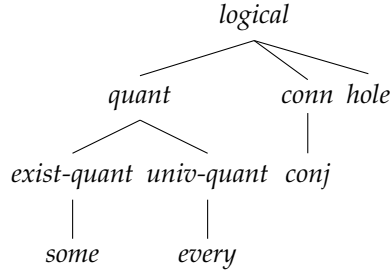

\begin{fs}
\begin{scope}
\Tree[.\type{logical} [.\type{quant}  [.\type{exist-quant} \type{some} ]
                            [.\type{univ-quant} \type{every}  ] ]
                [.\type{conn} \type{conj} ]
                \type{hole} ]
\end{scope}
\end{fs}
\caption{Part of a quantifier type hierarchy. In this Hase diagram an edge from a lower node to a higher one indicates an implication constraint from the first to the second.}
\label{hier}
\end{figure}

Cascades can also be stated. Equation \eqref{eq:casc} provides a sample of relevant formulas. But some constraints, like the transitivity of dominance edges \type{scope} do not fall under the format of one-place Horn-predicates. Nevertheless, they do not seem to raise any issues, so they may be also added. Last but not least, some structure might be very useful to avoid weird behaviours. As homomorphisms allow the transformation of nodes into wrappings, every such mapping would not be relevant. Looking back to Fig.~\ref{fs:bark-dog} on the right, the only interesting ``wrappingization'' is only these of \type{hole} nodes, i.e. the \feat{restr} or \feat{nscope}. It corresponds to the solving of the underspecified representation. Permitting only this phenomenon would require a schema: $\bigwedge_{\type{t}\neq\type{hole}}\forall T.~(\WR(T)\wedge T\bl\type{t})\to\bot$ though.

\begin{equation}
\label{eq:casc}
\begin{array}{l}
\forall e.~e\bl (\type{quant}\wedge\feat{restr}:\top)~\to~e\bl\feat{restr}:hole\\
\forall e.~e\bl (\type{quant}\wedge\feat{nscope}:\top)~\to~e\bl\feat{nscope}:hole\\
\end{array}
\end{equation}

However, constraints are not adequate to solve scope ambiguities alone because of their local character. To do so we need an entirely different process.

%%%
\subsection{Interaction with the instance}
\label{sub:inst}

Our final goal is to be able to connect the sentential meaning with the current knowledge of the world. This brings us beyond the sentential level, to the discourse level. While a sentence is represented, as discussed so far, by a Quantified Complex, the current world (the context) is just represented as usual (non-extended) frames. It is called the instance. However, it does not put up with ambiguity, i.e. it is unequivocal. So the underspecified representation of a sentence has to let us generate every possible reading of it.

Given that Quantified Complexes are very syntactic-like and close to Dominance Graphs, the solution is to utilize directly the \textbf{graph solver} of \cite{KollerThater:05}. The idea is to pick non-deterministically free fragments (nodes not having ingoing scope relations) and attach the recursively solved parts of the Weakly Connected Components created by the erasing of the selected fragment to its arguments (\type{hole} nodes). This forms new frames where every hole is mapped to a logical node or a wrapping, respecting scope constraints. Given that we would like this process to be monotonic w.r.t. $\sqsubseteq$, it explains why we allow nodes to be mapped to wrappings in homomorphisms. Algorithm \ref{algo:solver} explains the solving procedure. We refer to \cite{KollerThater:05} for further details on how this algorithm works and especially in which cases. Especially with donkey sentences (see footnote \ref{fn:donkey}), it fails to give a solution whereas our underspecified representation has the advantage to render this phenomenon properly, as in Fig.~\ref{fs:donkey}. See Fig.~\ref{fs:stroking} for a simple example with \ref{ex:stroking}. Remark that every \type{hole} node is map

\begin{figure}[ht]
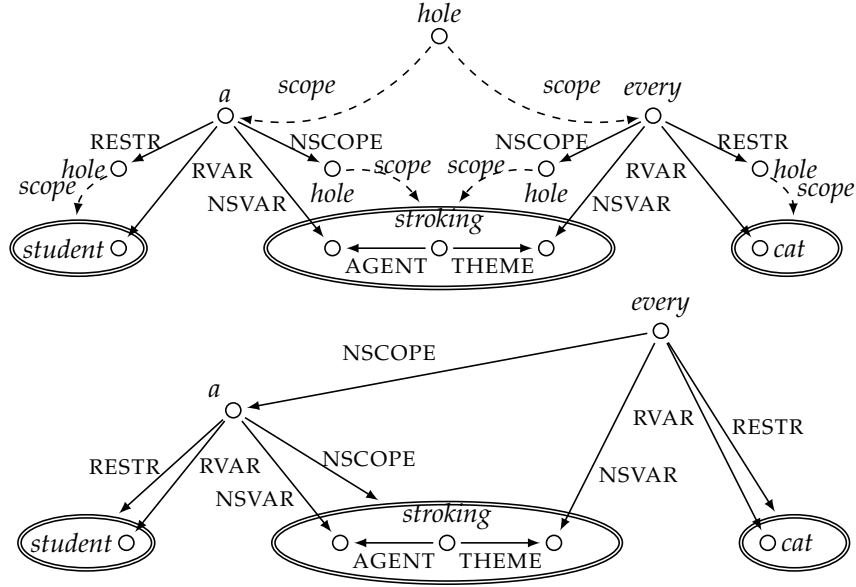

\begin{minipage}{\textwidth}
\begin{fs}
\draw (-15pt, 0pt)node[ft,minimum width=50pt,minimum height=20pt](ft1){};
\begin{scope}
    \draw (0pt, 0pt)node[nod,label={180:\type{student}}](11){};
\end{scope};

\draw (120pt, 0pt)node[ft,minimum width=130pt,minimum height=30pt](ft0){};
\begin{scope}[xshift=80]
    \draw (0pt, 0pt)node[nod](01){};
    \draw (40pt, 0pt)node[nod,label={90:\type{stroking}}](00){};
    \draw (80pt, 0pt)node[nod](02){};
    \path[->] (00) edge node[below]{\feat{agent}}(01);
    \path[->] (00) edge node[below]{\feat{theme}}(02);
\end{scope};

\draw (250pt, 0pt)node[ft,minimum width=40pt,minimum height=20pt](ft2){};
\begin{scope}[xshift=240]
    \draw (0pt, 0pt)node[nod,label={0:\type{cat}}](22){};
\end{scope};

\draw (40pt, 50pt)node[nod,label={90:\type{a}}](so){};
\draw (0pt, 30pt)node[nod,label={180:\type{hole}}](sor){};
\draw (80pt, 30pt)node[nod,label={-90:\type{hole}}](sons){};
\path[->] (so) edge node[left]{$\feat{restr}$}(sor);
\path[->] (sor) edge[bend right,dashed] node[left]{$\type{scope}$}(ft1);
\path[->] (so) edge node[right,pos=0.4]{$\feat{rvar}$}(11);
\path[->] (so) edge node[right]{$\feat{nscope}$}(evns);
\path[->] (sons) edge[bend left,dashed] node[right,pos=0.2]{$\type{scope}$}(ft0);
\path[->] (so) edge node[left,pos=0.7]{$\feat{nsvar}$}(01);
\draw (200pt, 50pt)node[nod,label={90:\type{every}}](ev){};
\draw (240pt, 30pt)node[nod,label={0:\type{hole}}](evr){};
\draw (160pt, 30pt)node[nod,label={-90:\type{hole}}](evns){};
\path[->] (ev) edge node[right]{$\feat{restr}$}(evr);
\path[->] (evr) edge[bend left,dashed] node[right]{$\type{scope}$}(ft2);
\path[->] (ev) edge node[left,pos=0.4]{$\feat{rvar}$}(22);
\path[->] (ev) edge node[left]{$\feat{nscope}$}(evns);
\path[->] (evns) edge[bend right,dashed] node[left,pos=0.2]{$\type{scope}$}(ft0);
\path[->] (ev) edge node[right,pos=0.7]{$\feat{nsvar}$}(02);
\draw (120pt, 80pt)node[nod,label={90:\type{hole}}](h){};
\path[->] (h) edge[bend left,dashed] node[above left]{$\type{scope}$}(so);
\path[->] (h) edge[bend right,dashed] node[above right]{$\type{scope}$}(ev);
\end{fs}

\begin{fs}
\draw (-15pt, 0pt)node[ft,minimum width=50pt,minimum height=20pt](ft1){};
\begin{scope}
    \draw (0pt, 0pt)node[nod, label={180:\type{student}}](11){};
\end{scope};

\draw (120pt, 0pt)node[ft,minimum width=130pt,minimum height=30pt](ft0){};
\begin{scope}[xshift=80]
    \draw (0pt, 0pt)node[nod](01){};
    \draw (40pt, 0pt)node[nod, label={90:\type{stroking}}](00){};
    \draw (80pt, 0pt)node[nod](02){};
    \path[->] (00) edge node[below]{\feat{agent}}(01);
    \path[->] (00) edge node[below]{\feat{theme}}(02);
\end{scope};

\draw (250pt, 0pt)node[ft,minimum width=40pt,minimum height=20pt](ft2){};
\begin{scope}[xshift=240]
    \draw (0pt, 0pt)node[nod, label={0:\type{cat}}](22){};
\end{scope};

\draw (40pt, 50pt)node[nod, label={130:\type{a}}](so){};
\path[->] (so) edge node[left]{$\feat{restr}$}(ft1);
\path[->] (so) edge node[right,pos=0.4]{$\feat{rvar}$}(11);
\path[->] (so) edge node[right]{$\feat{nscope}$}(ft0);
\path[->] (so) edge node[left,pos=0.7]{$\feat{nsvar}$}(01);
\draw (200pt, 80pt)node[nod, label={90:\type{every}}](ev){};
\path[->] (ev) edge node[right]{$\feat{restr}$}(ft2);
\path[->] (ev) edge node[left,pos=0.4]{$\feat{rvar}$}(22);
\path[->] (ev) edge node[above left]{$\feat{nscope}$}(so);
\path[->] (ev) edge node[right,pos=0.7]{$\feat{nsvar}$}(02);
\end{fs}
\end{minipage}
\caption{At the top, underspecified frame for \ref{ex:stroking}. The labels are not drawn for legibility. At the bottom, the solved form with $\rel{every}>\rel{a}$. The \rel{scope} relations are not drawn, neither the \type{hole} types.}
\label{fs:stroking}
\end{figure}

Another incentive is to connect our system to traditional type-logical formulas. Indeed this allows us to give a second semantics to quantifier nodes, namely in term of logical quantifier (or with Generalized Quantifiers). This can be done during the solving or on the tree-shaped solved form. The principle of \textbf{formula transcription} consists in writing a formula with atoms, boolean connectives and (generalized) quantifiers which describes the way the instance verifying this utterance should be like. To remember which variables of the term are associated with which frame node, we set up a co-assignment function. Recursively we define $\FORM_{\gamma}(\tau)$ on a tree-shaped frame $\tau$ with $\gamma:V\rightharpoonup\X$ the co-assignment function and $\X$ a set of variables :

\begin{equation}
\label{eq:form}
\begin{array}{l}
\FORM_{\gamma}(W)=~\exists. \bar{y}.~(\bigwedge_{\beta\in\FACTS(W)}\beta~\wedge\bigwedge_{v\in\dom\gamma\cap W}k_v\cdot p_v\trieq \gamma(v)~)\\
\FORM_{\gamma}(e\bl (some\wedge \feat{restr}:\tau_1\wedge\feat{nscope}:\tau_2))=~\exists u.~\FORM_{\gamma_1}(\tau_1)\wedge\FORM_{\gamma_2}(\tau_2)\\
\FORM_{\gamma}(e\bl (every\wedge \feat{restr}:\tau_1\wedge\feat{nscope}:\tau_2))=~\forall u.~\FORM_{\gamma_1}(\tau_1)\to\FORM_{\gamma_2}(\tau_2)\\
\end{array}
\end{equation}

where $\bar{y}=L(W)\cap\NVAR$ is the set of variables labelling $W$, $k_v$ and $p_v$ are taken such that $\I_g(k_v)\in W$ and $\I_g(k_v,p_v)=v$, $\gamma_1=\gamma[u\mapsto e\bl\feat{rvar}]$ and $\gamma_2=\gamma[u\mapsto e\bl\feat{nsvar}]$. By $\FACTS(W)$ we mean a sufficient finite fact-set describing $W$, i.e. a set of atoms (except $\IN$, $\INWR$ and $\WR$) such that for every atom $\beta[k\bl p]$ satisfied in $W$ (just about the inside nodes) there exists a prefix $q$ of $p$ such that $\beta[k\bl q]\in\FACTS(W)$. These atoms are then easily translatable to a neo-Davidsionian predicate style. An example for Fig.~\ref{fs:bark-dog} (on ther right, taking the only reading) is given in \eqref{eq:bark-dog-form}, reminding \eqref{eq:interpr}.

\begin{equation}
\label{eq:bark-dog-form}
\forall u.~(\exists x_2.~x_2\bl\type{dog}\wedge x_2\trieq u)~\to~(\exists z_0,x_1.~z_0\bl\type{barking}\wedge z_0\bl\feat{agent}\trieq x_1\wedge\ x_1\trieq u)
\end{equation}

The advantage of having a formula is that we can use existing strategies to make the sentence interact with the instance. But the rising question is: what interaction do we want? On the first hand, it seems relevant to require \textbf{model checking}, i.e. tell whether the instance verifies the formula. That amounts to truth-conditional semantics. But to go further, we might require dynamic semantics, i.e. to \textbf{accommodate} (i.e. to update the model) so that it satisfies the formula. While the first procedure can be done with a FOL Model-Checker, the second one is more subtle. Namely, the use of quantifiers can lead to infinite creation of nodes in the general case. Trying to avoid this phenomenon in linguistic situations may use a maximal set of quantification (e.g. considering only the entities currently present with the speaker). However, this set is often implicit and it goes beyond our study. Moreover, this is still not clear when nodes have to be created, and when not. In simple examples like \ref{ex:stroking} with $\textit{every}>\textit{a}$, a new \type{stroking} node should be created for each \type{cat}. Only later could its \feat{agent} possibly merge with the individual whom we learn that the stroker of that cat (considered in that utterance) was. But this is not related to the sole presence of $\exists$ in the formulas, because it behaves differently in negative polarity contexts (in particular under negation, but it is less good defined for other quantifier scopes, like \textit{most}). So a completely separate analysis has to be made to be able to produce this algorithm.

%%%%%
\section{Conclusion and future prospects}

Feature Structures can be extended by wrappings, a new construction enabling us to consider certain sets of nodes as nodes themselves. Adding a wrapping constructor in Attribute-Value Logic preserves the translatability into First-Order Logic and interesting properties of model homomorphisms. This formalism is powerful enough to model quantifiers in frames, namely with a more syntactical viewpoint called Quantified Complexes. Moreover, it expands the formal pipeline of Kallmeyer and Osswald, by unifying the generated lexical minimal models to an underspecified sentential representation. Our system goes even further by proposing ways to deal with the discourse level.

Some questions are however still under discussion. The cognitive and linguistic interpretation of wrappings is the main one. A way to extend this research would be to work on the dynamic interaction of sentences with the instance, updating the current knowledge of the world. Developing wrappings to model arbitrary sets and conjunction also seems to be promising. Other interests regarding frames range over the connection of this system with other close formalisms, like Cooper's Type Theory with Records for linguistics. 

\begin{small}
\begin{multicols}{2}

\end{multicols}
\end{small}

\begin{appendix}
%%%%%
\section{Additional figures}
\label{ap:fig}

\begin{figure}[ht]
\begin{minipage}{0.5\textwidth}
\begin{fs}
\draw (0,0) node[nod,label={180:$M_1:$}](11){};
\draw (40,0) node[nod](12){};
\path[->] (11) edge[above] node{\feat{p}} (12);
\begin{scope}[xshift=120]
\draw (0,0) node[nod,label={180:$M_2:$}](21){};
\draw (40,0) node[nod](22){};
\draw (0,-30) node[nod](23){};
\path[->] (21) edge[above] node{\feat{p}} (22);
\end{scope}
\begin{scope}[xshift=120,yshift=-80pt]
\draw (0,0) node[nod,label={180:$M_3:$}](31){};
\draw (40,0) node[nod](32){};
\draw (0,-30) node[nod](33){};
\draw (40,-30) node[nod](34){};
\path[->] (31) edge[above] node{\feat{p}} (32);
\path[->] (33) edge[above] node{\feat{p}} (34);
\end{scope}
\begin{scope}[yshift=-80pt]
\draw (0,0) node[nod,label={180:$M_4:$}](41){};
\draw (40,-15) node[nod](42){};
\draw (0,-30) node[nod](43){};
\path[->] (41) edge[above] node{\feat{p}} (42);
\path[->] (43) edge[below] node{\feat{p}} (42);
\end{scope}

\path[->] (11) edge[hom,bend left] (21);
\path[->] (12) edge[hom,bend left] (22);
\path[->] (21) edge[hom,bend right] (31);
\path[->] (22) edge[hom,bend left] (32);
\path[->] (23) edge[hom,bend right] (33);
\path[->] (31) edge[hom,bend right] (41);
\path[->] (32) edge[hom,bend right] (42);
\path[->] (33) edge[hom,bend left] (43);
\path[->] (34) edge[hom,bend left] (42);
\path[->] (41) edge[hom,bend left] (11);
\path[->] (42) edge[hom,bend right] (12);
\path[->] (43) edge[hom,bend left] (11);
\end{fs}
\end{minipage}
\begin{minipage}{0.5\textwidth}
\begin{fs}
\begin{scope}[xshift=30,yshift=30]
\draw (0,0) node[nod,label={180:$M_1':$}](11){$x$};
\draw (40,0) node[nod](12){};
\path[->] (11) edge[above] node{\feat{p}} (12);
\end{scope}
\begin{scope}[xshift=120]
\draw (0,0) node[nod,label={180:$M_2':$}](21){$x$};
\draw (40,0) node[nod](22){};
\draw (0,-30) node[nod](23){$y$};
\path[->] (21) edge[above] node{\feat{p}} (22);
\end{scope}
\begin{scope}[xshift=120,yshift=-80pt]
\draw (0,0) node[nod,label={180:$M_3':$}](31){$x$};
\draw (40,0) node[nod](32){};
\draw (0,-30) node[nod](33){$y$};
\draw (40,-30) node[nod](34){};
\path[->] (31) edge[above] node{\feat{p}} (32);
\path[->] (33) edge[above] node{\feat{p}} (34);
\end{scope}
\begin{scope}[yshift=-80pt]
\draw (0,0) node[nod,label={180:$M_4':$}](41){$x$};
\draw (40,-15) node[nod](42){};
\draw (0,-30) node[nod](43){$y$};
\path[->] (41) edge[above] node{\feat{p}} (42);
\path[->] (43) edge[below] node{\feat{p}} (42);
\end{scope}
\begin{scope}[xshift=-30,yshift=-30]
\draw (0,0) node[nod,label={180:$M_5':$}](51){$x,y$};
\draw (40,0) node[nod](52){};
\path[->] (51) edge[above] node{\feat{p}} (52);
\end{scope}

\path[->] (11) edge[hom,bend left] (21);
\path[->] (12) edge[hom,bend left] (22);
\path[->] (21) edge[hom,bend right] (31);
\path[->] (22) edge[hom,bend left] (32);
\path[->] (23) edge[hom,bend right] (33);
\path[->] (31) edge[hom,bend right] (41);
\path[->] (32) edge[hom,bend right] (42);
\path[->] (33) edge[hom,bend left] (43);
\path[->] (34) edge[hom,bend left] (42);
\path[->] (41) edge[hom,bend left] (51);
\path[->] (42) edge[hom,bend right] (52);
\path[->] (43) edge[hom,bend left] (51);
\end{fs}
\end{minipage}
\caption[caption]{\\
$\bullet$ On the left: Four frames without reachability constraint of a same equivalence class and homomorphisms in gray.\\
$\bullet$ On the right: Five linear ordered frames with reachability constraint. They are each time (not unique) successors of one another.\\
$\bullet$ The key is to be able to distinguish $M_1'$ from $M_5'$, i.e. make the difference between a single simple node and a single multiple node (stemming from the merging of two single nodes), which hence contains inherently strictly more information.}
\label{fs:eq-class}
\end{figure}

\begin{figure}[ht]
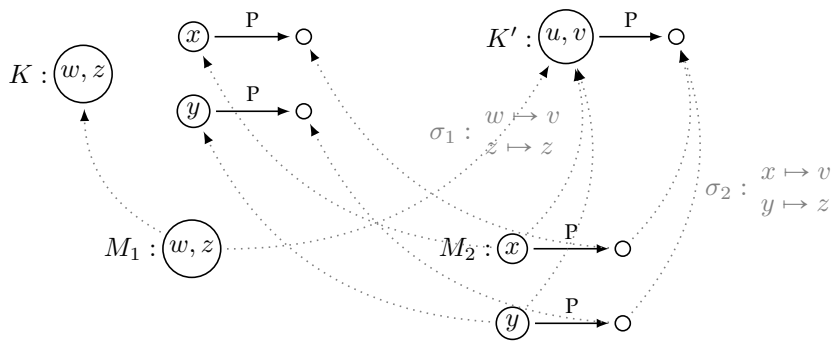

\begin{fs}
\begin{scope}[xshift=-20]
\draw (-20,-15) node[draw,circle,label={180:$K:$}](30){$w,z$};
\draw (20,0) node[draw,circle](31){$x$};
\draw (60,0) node[draw,circle](32){};
\draw (20,-30) node[draw,circle](33){$y$};
\draw (60,-30) node[draw,circle](34){};
\path[->] (31) edge[above] node{\feat{p}} (32);
\path[->] (33) edge[above] node{\feat{p}} (34);
\end{scope}
\begin{scope}[xshift=140]
\draw (0,0) node[draw,circle,label={180:$K':$}](11){$u,v$};
\draw (40,0) node[draw,circle](12){};
\path[->] (11) edge[above] node{\feat{p}} (12);
\end{scope}
\begin{scope}[xshift=120,yshift=-80pt]
\draw (0,0) node[draw,circle,label={180:$M_2:$}](21){$x$};
\draw (40,0) node[draw,circle](22){};
\draw (0,-30) node[draw,circle](23){$y$};
\draw (40,-30) node[draw,circle](24){};
\path[->] (21) edge[above] node{\feat{p}} (22);
\path[->] (23) edge[above] node{\feat{p}} (24);
\end{scope}
\begin{scope}[yshift=-80pt]
\draw (0,0) node[draw,circle,label={180:$M_1:$}](51){$w,z$};
\end{scope}

\path[->] (51) edge[draw=gray,dotted,bend left] (30);
\path[->] (51) edge[draw=gray,dotted,bend right] node[pos=0.8]{$\gray{\sigma_1:}\begin{array}{l}\gray{w\mapsto v}\\\gray{z\mapsto z}\end{array}$}(11);
\path[->] (21) edge[draw=gray,dotted,bend right] (11);
\path[->] (21) edge[draw=gray,dotted,bend left] (31);
\path[->] (22) edge[draw=gray,dotted,bend right] (12);
\path[->] (22) edge[draw=gray,dotted,bend left] (32);
\path[->] (24) edge[draw=gray,dotted,bend right] node[right]{$\gray{\sigma_2:}\begin{array}{l}\gray{x\mapsto v}\\\gray{y\mapsto z}\end{array}$}(12);
\path[->] (24) edge[draw=gray,dotted,bend left] (34);
\path[->] (23) edge[draw=gray,dotted,bend right] (11);
\path[->] (23) edge[draw=gray,dotted,bend left] (33);
\end{fs}
\caption{Example of two models having no least upper bound (lub) up to variable renaming. The $\alpha$-subsumed model on the left ($K$) would be the wanted result. But $M_1,M_2\sqsubseteq K'$ without having $K\sqsubseteq K'$.}
\label{fs:lub}
\end{figure}

\begin{equation}
\label{eq:abstr-fact}
\begin{array}{rl}
U_1 :&\begin{array}{c}
\IN(x,T_1), x\bl\feat{pp}\trieq x\bl\varepsilon, b\bl\feat{q}\trieq T_3\bl\varepsilon
,\IN(y,T_2), y\bl\feat{p}:t,\rel{r}(b,y), x\bl\feat{p}\trieq y\bl\varepsilon,\IN(z,T_3), z\bl\feat{q}:\top\\
x\bl\feat{pp}:\top,x\bl \feat{p}:\top,x\bl\top,\INWR(x),\WR(T_1),\INWR(T_2),y\bl\feat{p}:\top,y\bl\top,b\bl\feat{q}:\top,b\bl\top,\\
\INWR(z),\WR(T_3),z\bl\top,T_1\trieq T_2,\IN(x,T_2),\IN(y,T_1)\\
y\bl\feat{pp}\trieq y\bl\varepsilon,x\bl\feat{pp}:t\\[1.5em]
\end{array}\\
U_2 :&\begin{array}{c}
\IN(x,T_1), x\bl\feat{pp}\trieq x\bl\varepsilon, b\bl\feat{q}\trieq T_3\bl\varepsilon
,\IN(y,T_2), y\bl\feat{p}:t,\rel{r}(b,y), x\bl\feat{p}\trieq y\bl\varepsilon,\IN(z,T_3), z\bl\feat{q}:\top\\
x\bl\feat{pp}:\top,x\bl \feat{p}:\top,x\bl\top,\INWR(x),\WR(T_1),\INWR(T_2),y\bl\feat{p}:\top,y\bl\top,b\bl\feat{q}:\top,b\bl\top,\\
\INWR(z),\WR(T_3),z\bl\top,T_1\trieq T_2,\IN(x,T_2),\IN(y,T_1)\\
y\bl\feat{p}\trieq y\bl\varepsilon,x\bl\feat{pp}:t,x\bl\feat{pp}\trieq x\bl\feat{p},x\bl\feat{p}\trieq x\bl\varepsilon,x\bl\feat{p}:t\\[1.5em]
\end{array}\\
&\text{Sufficient fact-sets for \eqref{eq:abstr} produced from \eqref{eq:abstr-cnf}.}\\
&\text{The symmetric atoms are not written.}
\end{array}
\end{equation}

\begin{figure}[ht]
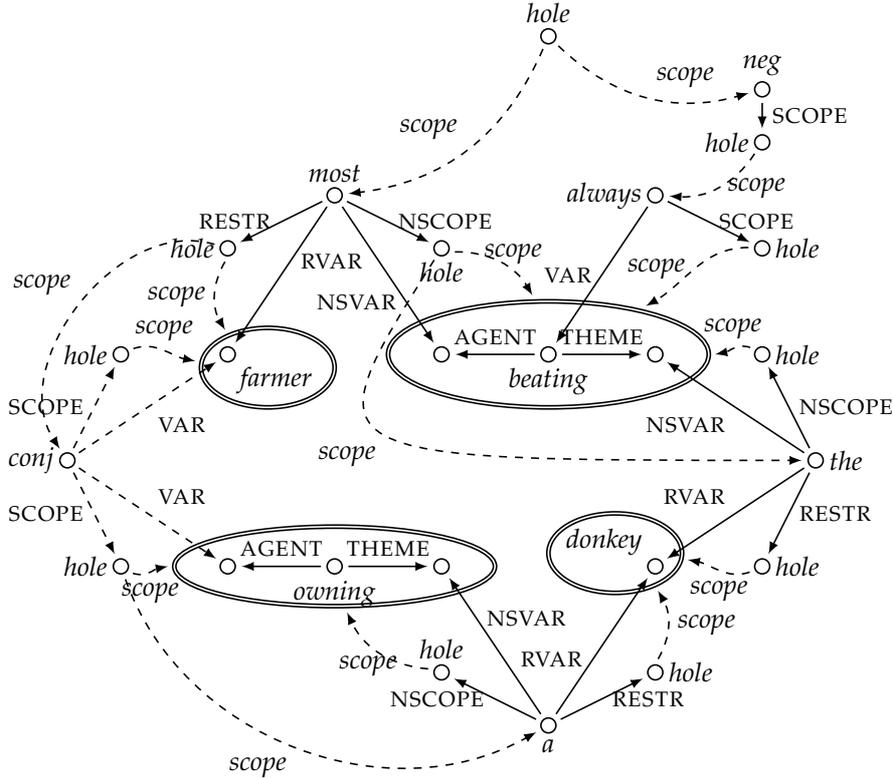

\begin{fs}
\draw (15pt, -5pt)node[ft,minimum width=50pt,minimum height=30pt](ft1){};
\begin{scope}
    \draw (0pt, 0pt)node[nod,label={-40:\type{farmer}}](11){};
\end{scope};
\draw (120pt, 0pt)node[ft,minimum width=120pt,minimum height=40pt](ft0){};
\begin{scope}[xshift=80]
    \draw (0pt, 0pt)node[nod](01){};
    \draw (40pt, 0pt)node[nod,label={-90:\type{beating}}](00){};
    \draw (80pt, 0pt)node[nod](02){};
    \path[->] (00) edge node[above]{\feat{agent}}(01);
    \path[->] (00) edge node[above]{\feat{theme}}(02);
\end{scope};
\draw (145pt, -75pt)node[ft,minimum width=50pt,minimum height=30pt](ft2){};
\begin{scope}[xshift=160,yshift=-80]
    \draw (0pt, 0pt)node[nod,label={150:\type{donkey}}](22){};
\end{scope};
\draw (40pt, -80pt)node[ft,minimum width=120pt,minimum height=30pt](ft3){};
\begin{scope}[xshift=0,yshift=-80]
    \draw (0pt, 0pt)node[nod](31){};
    \draw (80pt, 0pt)node[nod](32){};
    \draw (40pt, 0pt)node[nod,label={-90:\type{owning}}](30){};
    \path[->] (30) edge node[above]{\feat{theme}}(32);
    \path[->] (30) edge node[above]{\feat{agent}}(31);
\end{scope};
\draw (40pt, 60pt)node[nod,label={90:\type{most}}](mo){};
\draw (0pt, 40pt)node[nod,label={180:\type{hole}}](mor){};
\draw (80pt, 40pt)node[nod,label={-90:\type{hole}}](mons){};
\path[->] (mo) edge node[left]{$\feat{restr}$}(mor);
\path[->] (mor) edge[bend right,dashed] node[left]{$\type{scope}$}(ft1);
\path[->] (mo) edge node[right,pos=0.4]{$\feat{rvar}$}(11);
\path[->] (mo) edge node[right]{$\feat{nscope}$}(mons);
\path[->] (mons) edge[bend left,dashed] node[right,pos=0.2]{$\type{scope}$}(ft0);
\path[->] (mo) edge node[left,pos=0.7]{$\feat{nsvar}$}(01);

\draw (220pt, -40pt)node[nod,label={0:\type{the}}](th){};
\draw (200pt, -80pt)node[nod,label={0:\type{hole}}](thr){};
\draw (200pt, 0pt)node[nod,label={0:\type{hole}}](thns){};
\path[->] (th) edge node[right]{$\feat{restr}$}(thr);
\path[->] (thr) edge[bend left,dashed,in=180] node[below]{$\type{scope}$}(ft2);
\path[->] (th) edge node[above left]{$\feat{rvar}$}(22);
\path[->] (th) edge node[right]{$\feat{nscope}$}(thns);
\path[->] (thns) edge[bend right,dashed,in=180] node[above]{$\type{scope}$}(ft0);
\path[->] (th) edge node[below left]{$\feat{nsvar}$}(02);

\draw (-60pt, -40pt)node[nod,label={180:\type{conj}}](co){};
\draw (-40pt, -80pt)node[nod,label={180:\type{hole}}](co1){};
\draw (-40pt, 0pt)node[nod,label={180:\type{hole}}](co2){};
\path[->] (co) edge[dashed] node[left]{$\feat{scope}$}(co1);
\path[->] (co1) edge[bend right,dashed,in=180] node[below]{$\type{scope}$}(ft3);
\path[->] (co) edge[dashed] node[above right]{$\feat{var}$}(31);
\path[->] (co) edge[dashed] node[left]{$\feat{scope}$}(co2);
\path[->] (co2) edge[bend left,dashed,in=180] node[above]{$\type{scope}$}(ft1);
\path[->] (co) edge[dashed] node[below right]{$\feat{var}$}(11);

\draw (120pt, -140pt)node[nod,label={-90:\type{a}}](so){};
\draw (160pt, -120pt)node[nod,label={0:\type{hole}}](sor){};
\draw (80pt, -120pt)node[nod,label={90:\type{hole}}](sons){};
\path[->] (so) edge node[right]{$\feat{restr}$}(sor);
\path[->] (sor) edge[bend right,dashed] node[right]{$\type{scope}$}(ft2);
\path[->] (so) edge node[left,pos=0.4]{$\feat{rvar}$}(22);
\path[->] (so) edge node[left]{$\feat{nscope}$}(sons);
\path[->] (sons) edge[bend left,dashed] node[left,pos=0.2]{$\type{scope}$}(ft3);
\path[->] (so) edge node[right,pos=0.7]{$\feat{nsvar}$}(32);

\draw (120pt, 120pt)node[nod,label={90:\type{hole}}](h){};
\path[->] (h) edge[bend left,dashed] node[above left]{$\type{scope}$}(mo);
\draw[->,dashed] (mons) .. controls (30pt,-40pt) .. node[below left]{$\type{scope}$}(th);
\path[->] (mor) edge[bend right=70,dashed] node[above left]{$\type{scope}$}(co);
\path[->] (co1) edge[bend right=50,dashed] node[below left]{$\type{scope}$}(so);

\draw (200pt, 100pt)node[nod,label={90:\type{neg}}](ne){};
\draw (200pt, 80pt)node[nod,label={180:\type{hole}}](nes){};
\path[->] (h) edge[bend right,dashed] node[above right]{$\type{scope}$}(ne);
\path[->] (ne) edge node[right]{$\feat{scope}$}(nes);
\draw (160pt, 60pt)node[nod,label={180:\type{always}}](al){};
\draw (200pt, 40pt)node[nod,label={0:\type{hole}}](als){};
\path[->] (nes) edge[bend left,dashed] node[right]{$\type{scope}$}(al);
\path[->] (al) edge node[right]{$\feat{scope}$}(als);
\path[->] (als) edge[bend right,dashed,in=180] node[left]{$\type{scope}$}(ft0);
\path[->] (al) edge node[left]{$\feat{var}$}(00);
\end{fs}
\caption{Quantified Complex for the donkey sentence \textit{Most farmers who own a donkey don't always beat it}. General quantifiers like \textit{most} are modelled identically as logical ones. The relative clause is modelled as an adjunct connected by the conjunction node \type{conj} (with relations instead of attributes to be less syntactical). It is necessary to split the wrapping contents this way to derive correct readings. Quantifier \type{a} cannot rise above \type{most} because it is in \type{most} restrictor. The donkey anaphora \textit{it} is captured by the \type{the} representation, chosen because it also entails uniqueness and presupposition like an anaphora (\textit{it} = the donkey which was introduced), so scoping on the same \type{donkey} wrapping. We also propose here a modelling for other logical connectives (negation \type{neg}) and quantificational words (\textit{always}) following the traditional type-logical approach. One reading would be in modern GQT style: $\type{most}(x,\type{farmer}(x)\wedge\type{some}(y,\type{donkey}(y),\type{owning}(x,y)),\type{not}(\type{always}(e,\type{the}(y,\type{donkey}(y),\type{beating}(e,x,y)))))$}
\label{fs:donkey}
\end{figure}

\newpage
%%%%%
\section{Proofs}
\label{ap:proof}

\begin{theoreme}
\label{th:unique}
There is at most one homomorphism between two models.
\end{theoreme}

\begin{proof}
Take $M$ and $M'$ two models such that $h$ and $h'$ are two homomorphisms from $M$ to $M'$. Define $L(M)$ the set of all labels on which $\I_g$ (i.e. $\I$ and $g$) are defined. Then, for every node $v\in V$ such that (s.t.) there exists some $k\in L(M)$ s.t. $\I_g(k)=v$, we have $h(v)=h(\I_g(k))=\I'_{g'}(k)=h'(I_g(k))=h'(v)$. Now, for any node $v\in V$ reachable from $(v',p)$, by a straightforward induction on $p$, $h(v)=\I'(p)(h(v'))=\I'(p)(h'(v'))=h'(v)$. Finally, for a wrapping $W\in \W$, there is at least one node $v\in W$ labelled, by say $k$. So $h(v)=h(v')$ and it belongs to $h(W)$ and $h'(W)$, which must be hence equal thanks to the partition $\hat{\W'}$. Therefore $h=h'$.
\end{proof}

\begin{lemma}
\label{lem:mei}
Mutual subsumption is equivalent to isomorphism.
\end{lemma}

\begin{proof}
Take $h:M\to M'$ and $h':M'\to M$ two mutual homomorphisms. Define $\sigma=h'\circ h:M\to M$. As we also have $\mathsf{id}\eqdef\ <\!\mathsf{id}_V,\mathsf{id}_{\W}\!>:M\to M$ as obvious homomorphism, by lemma \ref{th:unique} $\sigma=\mathsf{id}$. So, if we had for some $v\in V$, $h(v)\in\W'$, it would imply $\sigma(v)\in\W$, what is wrong. So $\codom~h_V\subseteq V'$. Moreover, we directly have the bijectivity of $h$ by the surjectivity of $\sigma$ and by symmetry. Homomorphisms $h$ and $h'$ also tell us that $M$ and $M'$ have the same labels thanks to (base-labels) and (assignments) of \eqref{eq:hom}, i.e. $L(M)=L(M')$, so by reusing the proof of lemma \ref{lem:mei} (in fact properties (types) and (relations) of \eqref{eq:hom} are not necessary in the lemma), $h^{-1}=h'$. So $h^{-1}$ is a homomorphism, so $h$ is an isomorphism.
\end{proof}

\begin{theoreme}
\label{th:order}
Subsumption is an order up to isomorphism.
\end{theoreme}

\begin{proof}
Reflexivity is trivial with any isomorphism.

Transitivity is straightforward and left to the reader.

Symmetry results from proposition \ref{th:order}.
\end{proof}

\begin{theoreme}[Corollary]
\label{th:lub}
If two models $M$ and $M'$ have two least upper bound $K$ and $K'$, then they are isomorphic ($K\simeq K'$).
\end{theoreme}

\begin{proof}
By definition of the lub, $M,M'\sqsubseteq K,K'$ and for all $K''$ such that $M,M'\sqsubseteq K''$, we have $K',K\sqsubseteq K''$. In particular, $K\sqsubseteq K'$ and $K'\sqsubseteq K$, so by lemma \ref{lem:mei}, $K$ and $K'$ are isomorphic.
\end{proof}

\begin{theoreme}
\label{th:cmms}
If $\M_1$ and $\M_2$ are two canonical Minimal Model Sets for $\chi$, then $\M_1\simeq \M_2$.
\end{theoreme}

\begin{proof}
By hypothesis $\chi\models\M_2$, we have that for every model $M_2'$ such that $M_2'\models\chi$, there exists $M_2\in\M_2$ such that $M_2'\models M_2$. Yet every $M_1\in\M_1$ is a candidate by hypothesis $\M_1\models\chi$. So by setting $M_1\in\M_1$, there exists $M_2\in\M_2$ such that $M_2\sqsubseteq M_1$. By symmetry there is also some $M_2'\in\M_2$ such that $M_1\sqsubseteq M_2'$. As $\M_2'$ is an antichain, we must have $M_2=M_2'$, so $M_2\simeq M_1$. Therefore by symmetry, $\M_1$ and $\M_2$ are bijective sets of isomorphic models.
\end{proof}

\begin{theoreme}
\label{th:reduc}
If $U$ and $U'$ are $C_0$-saturated fact-sets, $U\models U'$ iff $U'\subseteq U$.
\end{theoreme}

\begin{proof}
By running through the definition of $U\models U'$ (conjunctive sets) we get $U\models U'$ iff $\forall\beta'\in U',~ U\models\beta'$. Yet $U$ is $C_0$-saturated, so $\beta'$ belongs to $U$ (because every atom of $U'$ has less than $C_0$ attribute symbols), hence $U'\subseteq U$. Conversely if $U'\subseteq U$ then every $\beta'\in U'$ is in $U$, hence $U\models U'$.
\end{proof}

\section{The saturation algorithm}
\label{ap:satur}

The algorithm follows the canonical Horn-SAT algorithm, but creates a new fact-set for each consequent $\rho$ of the right part of a clause $C$ and adds deduced atoms in the meantime with $\FuncSty{deduce}$.

\IncMargin{1em}
\begin{algorithm}[H]
\DontPrintSemicolon
\SetKwData{List}{list}
\SetKwFunction{Fail}{fail}\SetKwFunction{Duplicate}{duplicate}\SetKwFunction{Deduce}{deduce}
\KwIn{$\xi$ : a conjunctive set of clauses of the form $C=\Sigma\Rightarrow \Delta$}
% \KwData{$\FuncSty{list}$: a function from the set of clauses to $\{\top,\bot\}$}
\KwOut{$\U$ a finite disjunctive set of finite saturated conjunctive sets of atoms}
\BlankLine
$\U\leftarrow \{\emptyset\}$\;% \{\{k\cdot\top~|~k\in L(\xi)\}\}$\;
\For{$i=1$ \textbf{to} $|\xi|$}{
    \For{$C=\Sigma\Rightarrow \Delta\in\xi$}{
        $\U'\leftarrow\emptyset$\;
        \For{$U\in\U$}{
            \If{$\Sigma\subseteq U$}{
                $\U'\leftarrow\U'\cup\{$\Deduce$(U\cup\rho)~|~\rho\in\Delta\}$\;
            }
        }
        $\U\leftarrow\U'$\
    }
}
\caption{\KwSty{Satur}\label{algo:satur}}
\end{algorithm}\DecMargin{1em}

with $L(\xi)$ the labels present in $\xi$. As we sometimes have to wait for the deduction to bring positive literals that can be matched with left parts of clauses, a simple propagation unit does not work, hence the outermost $\KwSty{for}$ loop. The function $\FuncSty{deduce}$ returns a saturated set with the schemata in \eqref{eq:deduce}. 

\begin{equation}
\label{eq:deduce}
\begin{array}{c}

\begin{array}{rlr}
k\bl p:t\implies& k\bl p: \top&\text{(path intrinsic def. by type)}\\
k\bl p\trieq l\bl q\implies& k\bl p: \top\wedge l\bl q:\top&\text{(path intrinsic def. by equality)}\\
k\bl p:\top\implies& k\bl q:\top& \text{for every prefix $q$ of $p$}\\&& \text{(prefix intrinsic def.)}\\
\rel{r}_m((k_i\bl p_i)_{i<m})\implies& \bigwedge_{i<m} k_i\bl p_i: \top&\text{(path intrinsic def. by relation)}\\
k\bl pp'\trieq l\bl qp'\implies& k\bl p\trieq l\bl q&\text{(prefix equality)}\\
k\bl p\trieq l\bl q\implies& l\bl q\trieq k\bl p&\text{(equality symmetry)}\\
k_1\bl p_1\trieq k_2\bl p_2\wedge k_2\bl p_2\trieq k_3\bl p_3\implies& k_1\bl p_1\trieq k_3\bl p_3&\text{(equality transitivity)}\\
k\bl p\trieq l\bl q\implies& \beta[k\bl p/l\bl q]&\text{if }\mea{\beta[k\bl p/l\bl q]}\leq C_0\text{ (substitution)}\\[1.5em]
\end{array}\\
\begin{array}{rlr}
\IN(k,l)\implies& k\bl\top\wedge l\bl\top& \text{(path intrinsic def. by $\IN$)}\\
\IN(k,l)\wedge k\bl p\trieq k'\bl \varepsilon\implies& \IN(k',l)&\text{(non-escapability 1)}\\
\IN(k,l)\wedge\IN(k',l')\wedge k\bl p\trieq k'\bl p'\implies& l\trieq l'&\text{(non-escapability 2)}\\
\IN(k,l)\wedge\IN(k,l')\implies& l\trieq l'&\text{(wrapping disjointedness)}\\[1.5em]
\INWR(k)\implies& k\bl\top& \text{(path intrinsic def. by $\INWR$)}\\
\WR(k)\implies& k\bl\top& \text{(path intrinsic def. by $\WR$)}\\
\IN(k,l)\implies& \INWR(k)\wedge\WR(l)&\text{($\IN$ consequences)}\\
\INWR(k)\wedge\WR(k)\implies& \bot&\text{(wrapping non-embeddedness)}\\
T\bl\top\implies& \WR(T)&\text{(syntactic type $\WVAR$)}\\
x\bl\top\wedge \WR(x)\implies& \bot&\text{(syntactic type $\NVAR$)}\\
b\bl\top\wedge \WR(b)\implies& \bot&\text{(syntactic type $\NNAME$)}
\end{array}
\end{array}
\end{equation}

Remark that the reachability constraint does not have to be added because it is forced by the syntax. Same for the non-emptiness of wrappings.

\section{The graph solver}

The graph solver (\cite[p.4-6]{KollerThater:05}) processes by computing recursively the arguments of fragments (attribute-reachability maximum subframe, e.g. a \type{quant} node and its \feat{restr} and \feat{nscope}) and attaching them (i.e. merging) at their \type{hole} argument nodes. A fragment is called free if there is no ingoing scope relation in the considered subgraph. They are computed by the \FuncSty{Free-Fragment} function. The Weakly Connected Components (WCCs) can be calculated in linear time. The whole enumeration also only need polynomial time.

\IncMargin{1em}
\begin{algorithm}[H]
\DontPrintSemicolon
\SetKwData{FREE}{free}
\SetKwBlock{Try}{try}{}\SetKwBlock{With}{with}{}
\SetKwFunction{FF}{Free-Fragments}\SetKwFunction{FAIL}{fail}\SetKwFunction{CHOOSE}{choose}
\SetKwFunction{WCCS}{WCCs}\SetKwFunction{GS}{Graph-Solver}
\KwIn{$G$ : a tree-shaped Quantified Complex}
% \KwData{}
\KwOut{$G'$ : a unambiguous correct reading of $G$}
\BlankLine
\FREE$\leftarrow$\FF$(G)$\;
\lIf{\FREE$=\emptyset$}{\FAIL}
\Else{\CHOOSE$(F\in~$\FREE)\;
    $G_1,...G_m\leftarrow$\WCCS$(G\setminus F)$\;
    \ForEach{$G_i\in G_1,...G_m$}{
        $G_i'\leftarrow$\GS$(G_i)$\;
    }
    $G'\leftarrow$ Attach $G_1',...G_m'$ at $F$\;
}
\Return{$G'$}
\caption{\FuncSty{Graph-Solver}$(G)$}\label{algo:solver}
\end{algorithm}\DecMargin{1em}

\end{appendix}
\end{document}